\documentclass[journal,comsoc]{IEEEtran}
\usepackage[T1]{fontenc}

\usepackage[normalem]{ulem}

\ifCLASSINFOpdf
\else
\fi
\usepackage[cmintegrals]{newtxmath}
	\usepackage{amsmath,amssymb,amsthm,array,hyperref,caption,subcaption,enumitem,tikz,bbm,tabu,tabularx}
\usepackage[short]{optidef}
\usepackage[ruled,vlined]{algorithm2e}

\newtheorem{defn}{Definition}
\newtheorem{theorem}{Theorem}
\newtheorem{prop}{Proposition}

\newenvironment{eq}{\vspace{-1em}\begin{equation*}}{\end{equation*}\vspace{-1em}}

\def\mbf{\mathbf}
\def\mbb{\mathbb}
\def\mc{\mathcal}
\def\mbs{\boldsymbol}
\def\mbbm{\mathbbm}

\def\yes{$\surd$}
\def\no{ }

\newcommand{\abs}[1]{\ensuremath{\left| #1 \right|}}

\newcommand{\subjto}{\mathop{\mathrm{s.t.}}}

\makeatletter
\newcommand{\leqnomode}{\tagsleft@true\let\veqno\@@leqno}
\newcommand{\reqnomode}{\tagsleft@false\let\veqno\@@eqno}
\makeatother

\usepackage[colorinlistoftodos,color=orange!70,textsize=scriptsize,shadow]{todonotes}
\usepackage{setspace}

\begin{document}
		%
		\title{Network Traffic Shaping for Enhancing\\Privacy in IoT Systems}
		%
		%
		%
		
		\author{Sijie~Xiong,~\IEEEmembership{Student~Member,~IEEE,}
			Anand~D.~Sarwate,~\IEEEmembership{Senior~Member,~IEEE,}
			and~Narayan~B.~Mandayam,~\IEEEmembership{Fellow,~IEEE}
			\thanks{Manuscript received March xx, 2020; revised xxxx xx, 2021.}%
			\thanks{This work was supported by the United States National Science
				Foundation under Grant numbers SaTC-1617849, CCF-1453432 and CRISP-1541069, and by DARPA and US Navy under contract N66001-15-C-4070. Any opinions, findings, and conclusions or recommendations expressed in this material are those of the author(s) and do not necessarily reflect the views of the National Science Foundation, DARPA, or the US Navy.
			}%
			\thanks{S. Xiong, A. D. Sarwate and N. B. Mandayam are with the Department of Electrical and Computer Engineering, Rutgers, The State University of New
				Jersey, Piscataway, NJ 08854 USA e-mail: sijie.xiong@rutgers.edu,
				anand.sarwate@rutgers.edu, narayan@winlab.rutgers.edu.}
		}
		
		%
		%

	\markboth{IEEE/ACM Transactions on Networking, under review, 2021}%
	{Xiong \MakeLowercase{\textit{et al.}}: Network Traffic Shaping for Enhancing Privacy in IoT Systems}
	%



	\maketitle
	
	\begin{abstract}
		Motivated by privacy issues caused by inference attacks on user activities in the packet sizes and timing information of Internet of Things (IoT) network traffic, we establish a rigorous event-level differential privacy (DP) model on infinite packet streams. We propose a memoryless traffic shaping mechanism satisfying a first-come-first-served queuing discipline that outputs traffic dependent on the input using a DP mechanism. We show that in special cases the proposed mechanism recovers existing shapers which standardize the output independently from the input. To find the optimal shapers for given levels of privacy and transmission efficiency, we formulate the constrained problem of minimizing the expected delay per packet and propose using the expected queue size across time as a proxy. We further show that the constrained minimization is a convex program. We demonstrate the effect of shapers on both synthetic data and packet traces from actual IoT devices. The experimental results reveal inherent privacy-overhead tradeoffs: more shaping overhead provides better privacy protection. Under the same privacy level, there naturally exists a tradeoff between dummy traffic and delay. When dealing with heavier or less bursty input traffic, all shapers become more overhead-efficient. We also show that increased traffic from a larger number of IoT devices makes guaranteeing event-level privacy easier. The DP shaper offers tunable privacy that is invariant with the change in the input traffic distribution and has an advantage in handling burstiness over traffic-independent shapers. This approach well accommodates heterogeneous network conditions and enables users to adapt to their privacy/overhead demands. 
	\end{abstract}
	
	
	%
	\IEEEpeerreviewmaketitle

	\section{Introduction}\label{sec:intro}
	
	%
	%
	%
	%
	\IEEEPARstart{P}{rivacy} is a crucial factor inhibiting the proliferation of IoT devices and systems. Privacy concerns are aggravated in applications such as smart home and smart healthcare where sensing data containing personal information is continuously generated and often transmitted wirelessly onto the cloud. The sheer volume of this data poses huge challenges for privacy protection and (network) resource management~\cite{porambage2016quest}.
	
	Privacy attacks for user information can happen to many forms of IoT data and motivate different kinds of countermeasures~\cite{lu2018internet}. For example, raw sensor data contains unique patterns which can be extracted by untrusted application servers for user activity inference. To counter inference attacks on raw sensor data, existing methods include but are not limited to \emph{obfuscation} (e.g., Raval et al.~\cite{raval2019olympus} and Malekzadeh et al.~\cite{malekzadeh2020privacy}) and \emph{quantization} (e.g., Xiong et al.~\cite{xiong2016randomized}). They aim to guarantee rigorous privacy and minimize data utility loss. In addition, encryption techniques protect the raw sensor data from network observers. However, even with encrypted data, the communication system itself may leak privacy. Network traffic of IoT devices contains identifiable information like packet headers, sizes and timing that are highly correlated with the underlying user activities. By means of eavesdropping on encrypted IoT network traffic to obtain such information, many recent works demonstrate successful traffic analysis attacks~\cite{tahaei2020rise} to recover private user information. 
	
	This paper addresses traffic analysis attacks on packet sizes and timing information of encrypted IoT network traffic. Examples include Apthorpe et al.~\cite{apthorpe2017closing} who use a clustering method on encrypted packet traces of smart home IoT devices to easily identify their operating states, with each state triggered by a specific user activity; Das et al.~\cite{das2016uncovering} show that encrypted traffic from wearable (e.g., a fitness tracker) Bluetooth Low Energy (BLE) signal allows a BLE sniffer to identify user activities (e.g., whether the person is at rest/working/walking/running) from packet-size and inter-arrival time distributions; Buttyan and Holczer~\cite{buttyan2012traffic} apply the Discrete Fourier Transform on traffic rates from a wireless body area sensor network to reveal the types of medical sensors mounted on the patient, hence the patient's health conditions.
	
	To mask private user information encoded in packet sizes and timing, \emph{network traffic shaping} is a standard technique to change the original packet sizes and timing by means of delaying, padding, fragmentation and inserting dummy packets. There have been extensive studies and designs of network traffic shaping in many contexts such as anonymous networks, traditional Internet, as well as IoT systems. However, designing a privacy-preserving traffic shaping mechanism $\mc{M}$ for IoT systems requires special tailoring. Here, we highlight the following unique characteristics of IoT networks~\cite{seliem2018towards} and their implications on privacy and resource requirements for $\mc{M}$,
	\begin{itemize}
		\item \textbf{Continuous sensing.} IoT devices continuously transmit sensing data in streams of packets, $\mc{M}$ should provide the \emph{same privacy guarantee at any time} during transmission.
		\item \textbf{Changing environment.} The network of IoT devices and their traffic features can change rapidly, $\mc{M}$ should guarantee privacy that is \emph{invariant} with the changes.
		\item \textbf{Limited resources.} Many IoT networks are low-bandwidth and certain applications are delay-sensitive. With resource constraints, $\mc{M}$ must be \emph{efficient} and incur minimal overhead in terms of delay and dummy traffic (padding and dummy packets) for privacy protection.
		\item \textbf{Heterogeneity.} Due to heterogeneous network conditions and user demands for privacy, $\mc{M}$ should offer \emph{tunable} privacy and optimize for the privacy-overhead tradeoff.
	\end{itemize}
	
	\subsection{Research Questions}
	In the face of traffic analysis attacks in existing IoT devices and networks that utilize packet sizes and timing information, what kind of privacy can we hope to achieve? How should we design new protocols, standards and $\mc{M}$ that not only enable privacy protection against such attacks, but also fulfill the preceding desiderata for IoT networks? How do different types of IoT network traffic impact the privacy-overhead tradeoff of $\mc{M}$? In this paper, we address these questions by setting up a packet stream model for IoT network traffic, developing a formal and tunable privacy model for protecting packet sizes and timing information. We then design a novel shaping mechanism under the established traffic and privacy models and formulate a problem to optimize its shaping overhead. In the end, we will demonstrate the performance of the shaper on different types of IoT traffic by comprehensive experiments.
	
	\begin{table*}[t]
		\centering
		\small
		\setlength\tabcolsep{5pt}
		\begin{tabular}[c]{|c|c|c|c|c|c|c|c|c|c|c|} 
			\hline
			\multicolumn{2}{|c|}{} & \multicolumn{3}{c|}{Anonymous Networks} & \multicolumn{2}{c|}{Traditional Internet} & \multicolumn{4}{c|}{IoT Networks}\\ \cline{3-11}
			\multicolumn{2}{|c|}{Network Traffic Shapers} & PM & CM & LP & Wright & Iacovazzi & Apthorpe & Alshehri & Xiong & Our \\
			\multicolumn{2}{|c|}{} &~\cite{chaum1981untraceable} &~\cite{kesdogan1998stop} &~\cite{fu2003analytical,wang2008dependent} & \cite{wright2009traffic} & \cite{iacovazzi2014internet} & \cite{apthorpe2019keeping} & \cite{alshehri2020attacking} & \cite{xiong2018defending} & Method \\
			\hline \hline
			\multicolumn{2}{|c|}{Privacy Measure}   & \multicolumn{3}{c|}{Entropy} & Accuracy    & MI       & Accuracy & $(\epsilon,\delta)$-DP & \multicolumn{2}{c|}{$\epsilon$-DP}\\
			\hline \hline
			& Invariant with prior change & \multicolumn{3}{c|}{\no}  & \no         &\no          & \no      & \multicolumn{3}{c|}{\yes} \\ \cline{2-11} 
			Privacy  & Same guarantee at any time  & \multicolumn{3}{c|}{\no}  & \no         &\no          & \no      & \multicolumn{3}{c|}{\yes} \\ \cline{2-11} 
			Guarantee & Packet sizes                & \yes    & \no   & \no      & \yes        & \yes        & \yes     & \yes      & \yes & \yes           \\ \cline{2-11} 
			& Packet timing               & \yes    & \yes  & \yes     & \no         & \yes        & \yes     & \no       & \no & \yes \\ 
			\hline \hline
			& Delay                       & \yes    & \yes  & \yes     & \no         & \yes        & \no      & \no       & \no & \yes           \\ \cline{2-11} 
			Overhead & Dummy traffic               & \yes    & \no   & \yes        & \yes        & \yes        & \yes     & \yes      & \yes & \yes           \\ \cline{2-11} 
			& Optimized for efficiency    & \no     &\no    & \no         & \yes        & \yes        & \no      & \no       & \yes & \yes \\
			\hline
		\end{tabular}
		\caption{Comparison between existing network traffic shapers and our method in terms of privacy and overhead.}
		\label{tab:related_work_cmp}
	\end{table*}
	
	\subsection{Related Work} \label{sec:related_work}
	Existing network traffic shaping mechanisms proposed in the contexts of anonymous networks, traditional Internet and IoT systems mainly differ in: i) privacy measure, ii) privacy guarantee (whether it is invariant with traffic changes, valid and tunable at all times during communication and whether the mechanism protects both packet sizes and timing information), and iii) overhead (what resources the mechanism utilizes for shaping and whether it's optimized for overhead efficiency). We center our discussion about related work in these aspects.
	
	\subsubsection{Anonymous Networks} 
	Traffic shaping in anonymous networks~\cite{diaz2004taxonomy} focus on hiding who is communicating with whom. An adversary can observe the correspondences between the size and timing of messages (packets) going into and out of an intermediate network node (e.g., a router), and make inferences about source-destination pairs in the network. Anonymity is measured by \emph{entropy} (i.e., the uncertainty of the observing adversary about the sender/recipient of the messages) and achieved by the design of \emph{mixes}. 
	
	Anonymity mixes include pool mixes (PM) such as the original Chaum~\cite{chaum1981untraceable} mix which collects a fixed number of messages, pads them to a uniform length and outputs them in a batch. It hides \emph{both the sizes and timing} of arrivals hence the correspondences between input and output messages. However, excessive padding and batching renders pool mixes \emph{inefficient} in terms of byte and delay overhead.
	
	To reduce latency, continuous mixes (CM) are proposed to delay individual messages by random amounts of time. For example, in the Stop-and-Go mix (SG-mix)~\cite{kesdogan1998stop}, random delays added to the messages follow an exponential distribution. SG-mix was later shown to provide maximum sender anonymity measured by entropy among continuous mixes~\cite{danezis2004traffic}, assuming Poisson arrivals. On top of delaying individual messages, link padding (LP) techniques such as independent~\cite{fu2003analytical} and dependent~\cite{wang2008dependent} link padding can further reduce latency and enhance anonymity by allowing insertion of dummy packets to the output traffic to match predefined transmission schedules. However, continuous mixes and link padding techniques are designed for timing-based traffic analysis attacks. By themselves, they protect \emph{only packet timing information} and assume that packets are already padded to the same size.
	
	In conclusion, mixes and link padding shape the input traffic to match either fixed or predefined transmission schedules. By changing the schedules, they can offer \emph{tunable} anonymity measured by entropy. Nonetheless, the act of padding all packets to the same size to avoid compromising anonymity makes these mechanisms inefficient in terms of byte overhead.
	
	
	\subsubsection{Traditional Internet}
	Traffic analysis attacks in traditional Internet, such as webpage identification in HTTPS connection~\cite{sun2002statistical}, spoken language detection~\cite{wright2007language} and conversation transcript reconstruction~\cite{white2011phonotactic} in encrypted VoIP calls, focus more on inferring private information within a single flow. This differs from identifying source-destination pairs in anonymous networks, and the subject of traffic privacy in this context often adopts different privacy measures.
	
	Wright et al.~\cite{wright2009traffic} and Iacovazzi and Baiocchi~\cite{iacovazzi2014internet} use \emph{accuracy of flow classification} as their privacy measure. The former proposes a countermeasure that hides \emph{only packet size information} by shaping the source packet-size distribution to look like a target distribution, while minimizing byte overhead. The latter extends the same idea but their design considers \emph{additional masking for packet timing information}. They also propose a partial masking algorithm where the tradeoff between privacy guarantee and masking cost is controlled by masking how much fraction of the entire traffic flow.
	
	\emph{Mutual information} (MI) between the original and shaped traffic is often used as a privacy measure for shaping mechanisms as well. For example, an ON-OFF shaping policy is developed by Feghhi et al.~\cite{feghhi2016proportional} to guarantee perfect privacy (0 MI) for \emph{packet timing information}. Mathur and Trappe~\cite{mathur2011bit} study the fundamental tradeoff in shaping mechanisms between delay, padding and level of unlinkability (measured by MI). They show that combining delay and padding can offer much higher privacy protection than either approach alone.
	
	\subsubsection{IoT Networks}
	Like traditional Internet, inference attacks on encrypted IoT traffic aim at recovering private information (e.g., user activities and health conditions, IoT device types and operating states, etc.) within a single packet stream. However, the aforementioned characteristics of IoT networks inspire the re-design of efficient traffic shaping mechanisms and the quest for more suitable privacy measures.
	
	Apthorpe et al.~\cite{apthorpe2019keeping} propose a stochastic traffic padding (STP) scheme to hide time periods of user activities in a smart home setting with reasonable padding overhead. They measure privacy by the \emph{accuracy of identifying user activity periods}, and trade off padding overhead with privacy by controlling the percentage of padding periods. Alshehri et al.~\cite{alshehri2020attacking} consider IoT device identification attacks based on typical packet-size sequences and propose padding packets with bytes of uniform random noise to guarantee \emph{approximate $(\epsilon,\delta)$-differential privacy} (DP)~\cite{dwork2006our}. Our earlier work~\cite{xiong2018defending} proposes a packet padding obfuscation mechanism satisfying \emph{pure $\epsilon$}-DP~\cite{dwork2006calibrating} and optimizes its padding overhead. The mechanism prevents a last-mile eavesdropper from inferring about IoT device types and operating states based on \emph{packet-size} distributions of encrypted network traffic.
	
	The novelty of this paper in the design of network traffic shaping scheme compared to the previous ones is twofold: i) it is the first shaper with \emph{pure $\epsilon$}-DP guarantees using packet fragmentation and queueing operations besides padding, dummy packets insertion and delaying and ii) it optimizes for overhead efficiency while providing tunable worst-case privacy guarantees for packet sizes as well as timing information which extends our previous work~\cite{xiong2018defending}.
	
	\subsection{Advantages of Differential Privacy} \label{sec:DP_ADV}
	
	Differential privacy~\cite{dwork2014algorithmic} has emerged over the last decade as a compelling framework for measuring the worst-case privacy risk in various applications. We believe that designing traffic shaping mechanisms with DP guarantee satisfies the various privacy requirements desired by IoT networks.
	
	Concretely, the framework of DP under continuous observation~\cite{dwork2010differential} allows us to develop a new privacy model on traffic shaping mechanisms that ensure the same worst-case privacy level at any time during transmission. On the contrary, the information-theoretic (entropy and MI) and accuracy-based privacy measures are inherently average measures at the packet stream level, and lack information about worst-case privacy loss~\cite{danezis2013measuring} at every time instance.  
	
	Information-theoretic and accuracy-based privacy measures also rely heavily on the prior distribution of the original traffic. Under these, shaping mechanisms designed and optimized for a particular source flow (e.g., SG-mix for Poisson arrivals) may perform poorly in terms of privacy guarantees~\cite{diaz2004comparison} on unstable and unpredictable real traffic. The privacy guarantees may also fall apart when attackers are equipped with other side information to update their prior beliefs. Conversely, DP does not depend on the prior distribution and has been shown to be resistant to arbitrary side information~\cite{kasiviswanathan2008note}. A shaping mechanism, if designed with DP, will guarantee the same level of privacy that is invariant with the changes in the original traffic or the attacker's side information.
	
	Lastly, DP offers directly tunable privacy parameters that can accommodate the heterogeneity in network conditions and user demands for privacy. Many of the other shaping mechanisms, however, require the choice of target traffic distributions to indirectly control their privacy guarantees.
	
	\subsection{Challenges}
	
	The listed advantages motivate us to design shaping mechanisms for encrypted IoT network traffic under the framework of DP. Table~\ref{tab:related_work_cmp} summarizes the differences between our proposed method and existing traffic shaping mechanisms. In order to counter traffic analysis attacks on packet sizes as well as timing information and to design overhead-efficient DP shapers, we face the following challenges,
	\begin{itemize}
		\item \textbf{Formal privacy model on encrypted IoT traffic.} We need to set up an appropriate IoT network traffic model, and more importantly develop a formal privacy model on shaping mechanisms with DP guarantees for protecting both packet sizes and timing information.
		\item \textbf{Easy-to-find shaper that is privacy-preserving and overhead-efficient.} We need to design a shaping mechanism that not only satisfies the formal privacy model but also consumes minimum overhead due to limited resources in IoT networks. Moreover, we want to find such privacy-preserving and overhead-efficient shaper easily.
	\end{itemize}
	\subsection{Contributions}
	
	Our work makes the following contributions to the subject of traffic privacy and countermeasure design in IoT networks,
	\begin{itemize}
		\item We set up an abstract discrete IoT packet stream model inspired by a smart home setting in Section~\ref{section:sys}. However, the methodology developed here after is generally applicable to other IoT applications. 
		\item We develop an event-level DP model on infinite packet streams in Section~\ref{sec:privacy_model}. The model defines tunable privacy guarantees for both packet sizes and timing information.
		\item In Section~\ref{sec:DPS}, we design a shaping mechanism under the event-level DP model. The mechanism satisfies a first-come-first-served (FCFS) queueing discipline. To reduce the shaping overhead needed for privacy protection, we allow the use of delay, fragmentation and dummy traffic to shape the outgoing packet stream leaving a local area IoT network (e.g., a smart home). In special cases, the mechanism recovers previous methods in the literature.
		\item We formulate the problem of finding event-level DP and overhead-efficient shaper in Section~\ref{sec:delay_opt_dps} as a constrained optimization: we minimize the expected queue size across time imposed by the shaper (as a proxy for the expected delay per packet) under given privacy and transmission efficiency levels. We further show that the optimal shaper can be easily found by convex programming.
		\item In Section~\ref{sec:EXP}, we conduct comprehensive evaluations on the privacy-overhead tradeoffs of different shapers on both synthetic data and packet traces from actual IoT devices. For IoT packet traces, we further show how shapers optimized with the assumption of independent and identically distributed (i.i.d.) network traffic perform differently on actual (probably bursty) traffic.
	\end{itemize}
	
	We believe that this work establishes a novel framework for building resource-efficient network traffic shaping systems with strong, formal and tunable privacy guarantee. The privacy-overhead tradeoffs carry meaningful indications for the design of future privacy-aware IoT systems. 
	
	\section{System Model}\label{section:sys}
	
	\begin{figure}[t]
		\centering
		\includegraphics[scale=0.35]{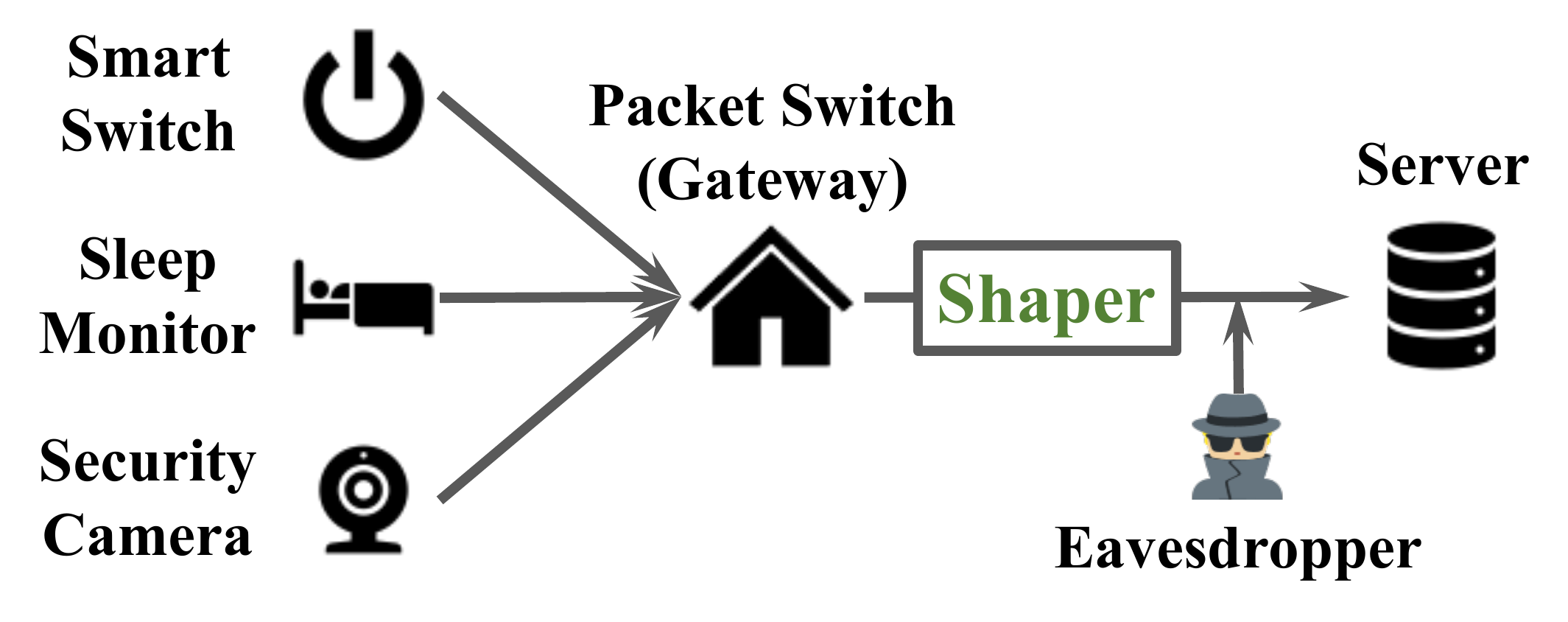}
		\caption{Illustration of a smart home traffic shaping system.}
		\label{fig:smart_home}
	\end{figure}
	
	\subsection{Overview}
	
	A prototypical realization of the IoT system~\cite{atzori2010internet} consists of a gateway system that manages networked IoT devices. For example, Fig.~\ref{fig:smart_home} shows a typical realization of the smart home IoT network. It consists of a WiFi access point (also acting as a packet switch/home gateway) that allows multiple heterogeneous monitoring devices to transmit information onto the application server in the form of data packets. 
	
	Each device has three modes of operation: sensing, update, and silence. In the sensing mode, a device can sense one or several types of user activities (or \emph{events}, e.g., a Nest camera can detect the motion of a user or whether the user is checking the camera feed) and subsequently transmit event-indicating traffic (e.g., a large packet or a short burst with distinctive size) to its application server. In the update mode, a device routinely sends packets containing status updates such as energy consumption levels and firmware versions. Lastly, devices send nothing in the silence mode. 
	
	We illustrate in Fig.~\ref{fig:agg_pkt_stream} a sample of aggregate packet stream during a 1-min time window as an input to the packet switch. Colored spikes indicate packet arrivals at different times. The smaller ones are status updates. The larger ones correspond to the event-triggered packets, for example, the 270B packet at 7s indicates that some user motion is detected by the camera and the 1117B packet at 30s marks the sleep onset of the user. The packet/burst sizes triggered by a particular type of event at various times are much alike and those triggered by different types of events are quite distinguishable~\cite{apthorpe2017smart}. Subahi and Theodorakopoulos~\cite{subahi2019detecting} survey and provide an extensive list of correspondences between user interactions with IoT mobile apps and the generated packet sizes/sequences.
	
	Denote $N^+=\{1,2,\ldots,n\}$ as the set of all observable events in an IoT network and $N\triangleq\{0\}\cup N^+$. When event $i\in N^+$ (e.g., a person going to sleep) happens, it triggers a device (e.g., a Sense Sleep monitor) to send an event packet of size $a_i>0$. Here, we model the event-indicating traffic as a single packet with distinctive size $a_i$ in order to abstract away from the details of event traffic\footnote{For example, if the event-triggered network traffic is a short burst, we can represent it by its burst size $a_i$. Additionally, if the packet/burst sizes triggered by a particular type of event change slightly at different occurring times, we can use the average packet/burst size to symbolize the event traffic.}. We also consider the aforementioned status updates as traffic generated by a special type of event to ignore the distinction. Let $i=0$ represent a \emph{null} event which stands for the silence mode, i.e., the absence of any events or updates. Denote $\mc{A}^+=\{a_1,a_2,\ldots,a_n\}$ as the set of all possible event packet sizes in bytes. We assume without loss of generality (w.l.o.g.) that $0<a_1<a_2<\cdots<a_n$. Let $\mc{A}\triangleq\{a_0\}\cup\mc{A}^+$ where $a_0=0$ represents packets of size zero in the case of null events, and $\mbs{a}\triangleq[a_0,a_1,\ldots,a_n]^\top$.  
	
	\begin{figure}[t]
		\centering
		\includegraphics[width=.85\linewidth,height=.35\linewidth]{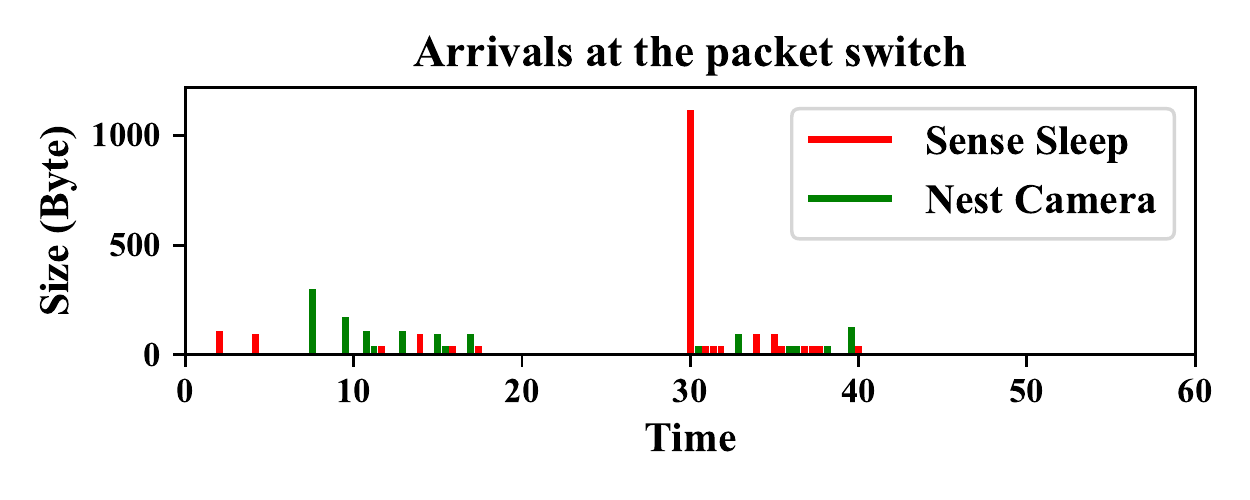}
		\caption{Aggregate packet stream arriving at the packet switch from Sense Sleep monitor (red) and Nest camera (green).}
		\label{fig:agg_pkt_stream}
	\end{figure}
	
	\subsection{Event Packet Stream} \label{sec:event_pkt_stream}
	
	We model the event packet stream arriving at the packet switch as a discrete-time\footnote{Time can be made discrete by specifying a time resolution. The transmission/processing times of variable-size packets differ, which can potentially be exploited by an eavesdropper. Here, we assume that this minute kind of timing information is destroyed by the inherent variability of network delay.} sequence $\{A_t\}$ of packets with variable sizes $A_1,A_2,\ldots$ drawn from the set $\mc{A}$. Time slots are indexed by the subscripts $t=1,2,\ldots$ and we assume one packet per slot (including zero packet). We use $A^T=(A_1,A_2,\ldots,A_T)$ to denote the $T$-\emph{prefix} of the event packet stream from the first to the $T$-th time slot. Let $\mc{T}\triangleq\{1,2,\ldots,T\}$ and denote $\mc{I}\triangleq\{t\in \mc{T}: A_t>0\}$ as the set of arrival times of event packets $A_t\in\mc{A}^+$ during time period $T$ to diffentiate from $\mc{T}$.
	
	We think of the packet sizes $A_t$ as i.i.d. with distribution $\lambda(a)$ on $\mc{A}$. Specifically, $\lambda_i\triangleq \lambda(a_i)\in(0,1)$ denotes the probability of event $i\in N^+$ (or the null event $i=0$) occurring in each time slot. It also measures the \emph{rate} at which $i\in N$ occurs over the infinite horizon, so that,
	\begin{align}
		\lambda_i=\lim_{T\rightarrow\infty}\frac{1}{T}\sum_{t=1}^T\mbbm{1}\{A_t=a_i\},\ \forall i\in N,
	\end{align}
	where $\mbbm{1}\{\cdot\}$ is the indicator function and $\sum_{i\in N} \lambda_i=1$. Let $\mbs{\lambda}\triangleq[\lambda_0,\lambda_1,\ldots,\lambda_n]^\top$. We can denote the \emph{arrival rate} of all event packets $A_{t\in\mc{I}}$ at the switch as,
	\begin{align}
		\Lambda = \sum_{i\in N^+}\lambda_i = 1-\lambda_0 \in (0,1) \label{eqn:arrival_rate}.
	\end{align}
	In addition, we measure the \emph{input byte rate} (expected number of bytes per time slot) to the packet switch as,
	\begin{align}
		B_{in} &= \mbb{E}_\lambda[A_t]=\sum_{i\in N^+}\lambda_i a_i=\mbs{\lambda}^\top\mbs{a}. \label{eqn:in_byte_rate}
	\end{align}
	
	It is useful to distinguish event packet streams based on their arrival and input byte rates. By an ``elephant'' (``mouse'') flow hereafter we mean an event packet stream with a high (\emph{resp.} low) arrival rate $\Lambda$ and input byte rate $B_{in}$. We then describe heavy traffic (e.g., aggregated from a larger number of IoT devices) as an elephant flow, and vice versa. In Section~\ref{sec:EXP}, we will show how mouse/elephant flows affect the privacy-overhead tradeoffs of our proposed shaping mechanism. 
	
	
	\subsection{Adversary Model}\label{sec:adversary_model}
	We assume that a last-mile eavesdropper in Fig.~\ref{fig:smart_home} observes the packet stream coming out of the gateway and is interested in identifying the ongoing user activities within the household. Due to packet encryption, they can only observe the timing and sizes of successive packets. We suppose that the adversary can also obtain the set $\mc{A}^+$ from other sources.
	
	In this work, we further assume an \emph{event-level} adversary who is interested in the \emph{type} and \emph{timing} of an event/activity. Concretely, the adversary wants to know whether event $i\in N^+$ or $j\in N^+$ (event type) happened given that they observed something in time slot $t$. For example, they can infer based on the transmitted packet size whether the user is checking the Nest camera's live feed (142B packet) vs being detected for motion (270B packet). The adversary can also discover the timing of an event based on the packet observing time. A 1117B packet observed at time $t$ (10pm) instead of $s$ (11pm) informs the adversary that the user is going to sleep at 10pm, since a 1117B packet is exclusively generated by the Sense Sleep monitor when it detects the user's sleep onset.
	
	Without traffic shaping, a simple receive-and-forward packet switch will output the exact same arrival $A_t$ which immediately gets exposed to the eavesdropper. This allows them to easily identify any event $\{i\in N: A_t=a_i\}$ happened at any time $t$ and fully uncover the ongoing user activities, violating privacy.
	
	\subsection{Shaping Mechanism} \label{sec:shaper_overview}
	
	To prevent the eavesdropper from inferring about private user activities, we design network traffic shaping mechanisms $\mc{M}$ at the packet switch to obfuscate the packet sizes and timing information. For ease of analysis, we let $\mc{M}:\mc{A}^T\rightarrow\mc{D}^T$ satisfy a FCFS queuing discipline. It takes as input a length-$T$ prefix of the event packet stream $A^T\in\mc{A}^T$ waiting to be served in order of arrival and outputs a same-length packet stream $D^T\in\mc{D}^T$ for arbitrary $T$. Here, $T$ can be specified beforehand as a fixed session time, or we can think of it as growing indefinitely. The shaped output $D^T$, which should convey less private user information, is also a sequence of random variables $D_{t\in\mc{T}}\in\mc{D}=\{d_0, d_1, d_2, \ldots, d_m\}$ where $\mc{D}$ is the set of output packet sizes in bytes and we assume w.l.o.g. that $0=d_0<d_1<d_2<\cdots<d_m$ with $d_m\geq a_n$. Let $M\triangleq\{0,1,2,\ldots,m\}$, $M^+\triangleq M\setminus\{0\}$, $\mc{D}^+\triangleq\mc{D}\setminus\{0\}$ and $\mbs{d}\triangleq[d_0,d_1,\ldots,d_n]^\top$.
	
	In this work, we let the shaping mechanism $\mc{M}$ perform ``surgery'' on packets from IoT devices such as splitting and reassembling, padding and delaying to map $A_t$ to $D_t$, where $D_t$ can be smaller than $A_t$. We assume an intermediate platform (e.g., a router connecting multiple local area IoT networks) that is trusted and shared by many households which can restore these manipulated packet streams before forwarding them to the intended application servers. In Section~\ref{sec:L&F}, we will discuss how this assumption could be relaxed and leave for future work the design of systems which can handle untrusted settings or the absence of such platform.
	
	Since $D_t$ disguises the true arrival $A_t$, the adversary becomes uncertain about the user activity instance. For example, if $A_t>0$ and $D_t=0$, the adversary would falsely believe that nothing has happened at time $t$. For the entire shaped packet stream $\{D_t\}$, we denote its \emph{output byte rate} as,
	\begin{align}
		B_{out} = \lim_{T\rightarrow\infty}\frac{1}{T}\sum_{t=1}^T D_t. \label{eqn:out_byte_rate}
	\end{align}
	We also let $Q_t$ be the size of the FCFS queue just after departure $D_t$. If $D_{t+1}>Q_t+A_{t+1}$, $\mc{M}$ consumes $D_{t+1}-A_{t+1}-Q_t$ amount of dummy bytes/packets to form the output $D_{t+1}$. Otherwise, $\mc{M}$ buffers the remaining bytes (if any) in the queue. Fig.~\ref{fig:shaper_diagram} illustrates the dynamics of the FCFS queue during shaping over 3 consecutive time slots when the departures have the same size $d$ for example. Packets $A_1$, $A_2$, $A_3$ encounter delays of $W_1=1$, $W_2=W_3=0$ time slots respectively and $3d-(A_1+A_2+A_3)$ dummy bytes are added for shaping.
	
	
	In light of resource constraints in IoT networks, we want to design $\mc{M}$ and minimize its shaping overhead while protecting privacy. To this end, we introduce importance overhead measures in the sequel. We will use them in the optimization of shapers later in Section~\ref{sec:OPT}, as well as in the performance comparisons between different shapers in Section~\ref{sec:EXP}.
	
	\begin{figure}[t]
		\centering
		\includegraphics[width=.75\linewidth,height=.35\linewidth]{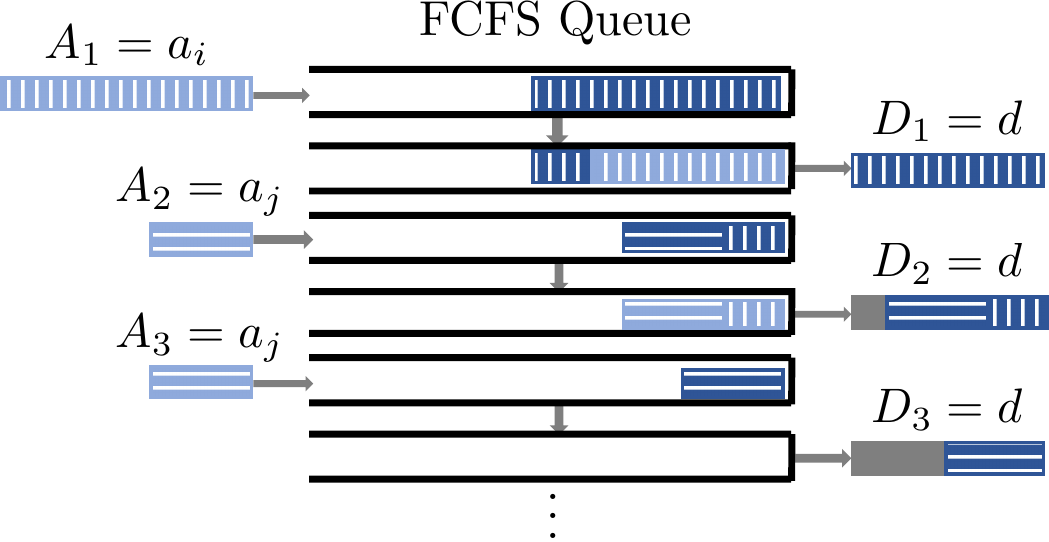}
		\caption{Illustration of the FCFS queue during shaping over 3 consecutive time slots with same-size departures. 
		}
		\label{fig:shaper_diagram}
	\end{figure}
	
	\subsection{Overhead Measures}\label{section:measures}
	
	\textbf{Transmission efficiency.} We define the transmission efficiency $\rho$ of shaping mechanisms as the total number of information bytes (i.e., bytes in the input event packet stream) divided by the total number of bytes in transmission (i.e., bytes in the shaped output packet stream) during time period $T$,
	\begin{align}
		\rho = \frac{B_{in}}{B_{out}},
	\end{align}
	where $B_{in}$~\eqref{eqn:in_byte_rate} and $B_{out}$~\eqref{eqn:out_byte_rate} are the byte rates of $\{A_t\}$ and $\{D_t\}$ respectively\footnote{The same multiplier $T$ in the numerator and denominator is cancelled out.}. The extra byte rate (dummy bytes per time slot) needed to shape $A^T$ to $D^T$ for arbitrary $T$ is $B_{out} - B_{in}=(1/\rho-1)B_{in}$. We say that a shaper is more transmission efficient (higher $\rho$) if it needs less dummy traffic (lower $1/\rho-1$) to shape the same input traffic (same $B_{in}$). 
	
	\textbf{Delay overhead.} We measure the delay overhead of a shaper as the average waiting time each event packet $A_{t\in\mc{I}}$ spends in the FCFS queue over the infinite horizon,
	\begin{align}
		\bar{W}_p=\lim_{|\mc{I}|\rightarrow\infty}\frac{1}{|\mc{I}|} \sum_{p=1}^{|\mc{I}|} W_p, \label{eqn:avg_delay_per_pkt}
	\end{align}
	where $p$ denotes the $p$-th event packet in the input packet stream, and $|\mc{I}|$ (the cardinality of $\mc{I}$) quantifies the total number of event packets during $T$, that $|\mc{I}|\rightarrow\infty$ as $T\rightarrow\infty$.
	
	\textbf{Queue size.} By the FCFS queueing discipline, the evolution of the queue size over time depicted in Fig.~\ref{fig:shaper_diagram} can be well described by the discrete-time version of Lindley's equation~\cite{lindley1952theory},
	\begin{align}
		Q_t=\max(Q_{t-1}+A_t-D_t,0),\quad \forall t\geq 1,\label{eqn:lindley_qsz}
	\end{align}
	assuming that we start with an empty queue $Q_0=0$. Define the average queue size across time as,
	\begin{align}
		\bar{Q}_t=\lim_{T\rightarrow\infty}\frac{1}{T} \sum_{t=1}^T Q_t. \label{eqn:avg_qsz_across_time}
	\end{align}
	
	It is well known from queueing theory that to prevent the queue from accumulating indefinitely, we need the following \emph{stability condition}~\cite{gallager2013stochastic},
	\begin{align}
		B_{in}<B_{out} \Leftrightarrow \rho\in(0,1). \label{eqn:stability}
	\end{align}
	This ensures for the shaping mechanism that both the average waiting time per packet and the average queue size across time converge to finite expectations ($\bar{W}_p\rightarrow\mbb{E}[W_p]<\infty$ and $\bar{Q}_t\rightarrow\mbb{E}[Q_t]<\infty$). The stability condition on the transmission efficiency $\rho$ in terms of the input and output byte rates is analogous to that on the server's \emph{utilization factor}~\cite{bertsekas1992data} in terms of customer arrival and departure rates.
	
	\subsection{Assumptions}
	
	To summarize, we make the following assumptions to create a model of the IoT traffic shaping system:
	\begin{itemize}
		\item The input packet stream $A^T$ consists of variable-sized packets and is i.i.d. across time. 
		\item The traffic shaper maps $A^T$ to an output packet stream $D^T$ by performing ``surgery'' on packets in $A^T$.
		\item The adversary observes the output packet stream $D^T$ to make inference about $A^T$.
	\end{itemize}
	In the following section, we will first present the definition of DP on data streams and formalize our privacy model for event packet streams. Particularly, our privacy guarantee will not depend on the i.i.d. assumption on $A^T$: they will hold for any realization of $A^T$. The i.i.d. assumption is only needed to construct convex programs to solve for the overhead-optimal shapers efficiently, which is of practical use. Moreover, we will demonstrate the advantage of our proposed shaper optimized under i.i.d. assumption when performed on real correlated traffic in Section~\ref{sec:exp_corr}. In Section~\ref{sec:shaping_mechanisms} onward we will look at shaping mechanisms with different guarantees according to the established DP model on event packet streams.
	
	\section{Privacy Model}\label{sec:privacy_model}
	
	As discussed in Section~\ref{sec:DP_ADV}, DP offers a quantifiable worst-case measure of privacy risk that is resistant to the change of prior distribution and the adversary's side information. We aim to design traffic shaping mechanisms $\mc{M}$ with DP guarantees to meaningfully and rigorously protect private information of continuously generated IoT network traffic and trade off privacy with shaping overhead. This section reviews the definition of DP on data streams. We then apply the definition to event packet streams to establish a formal privacy model on $\mc{M}$.
	
	\subsection{Differential Privacy on Data Streams} \label{section:stream_privacy}
	In the framework of DP, a trusted curator holds sensitive information from a group of users, creating a dataset $X\in\mc{X}$ where $\mc{X}$ denotes the universe of possible datasets. We think of the dataset $X$ as containing a set of rows with each row corresponding to a single user's data record. Two datasets $X, \tilde{X}\in\mc{X}$ are called \emph{neighboring} if they differ in a single row, i.e., a single user's data. 
	
	In the settings of continuous observation~\cite{dwork2010differential}, the dataset $X$ is updated by data streams, and the curator has to generate outputs continuously. A data stream is a time-ordered sequence of symbols $S_1,S_2,\ldots$ drawn from a domain $\mc{S}$, where each symbol represents an event. $S_t\in\mc{S}$ corresponds to the event happened in time slot $t$ and $S^T=(S_1,S_2,\ldots,S_T)$ denotes the $T$-prefix of the data stream. Entries in $X$ are therefore associated with events, or actions taken by the users. 
	
	Two differential privacy models -- \emph{event-level} and \emph{user-level} privacy~\cite{dwork2010differential} are proposed for data streams. The former hides a single event whereas the latter hides all the events of a single user. The difference between the two privacy models lies in the definition of \emph{adjacency} between stream prefixes. Specifically, $S^T$ and $\tilde{S}^T$ are event-level adjacent if there exists at most one time slot $t$ that $S_t\neq \tilde{S}_t$; $S^T$ and $\tilde{S}^T$ are user-level adjacent if for any user $S_t\neq \tilde{S}_t$ for arbitrary number of time slots. 
	
	Let $\mc{M}: \mc{S}^T\rightarrow \mc{O}$ denote a mechanism that takes as input a stream prefix $S^T\in\mc{S}^T$ of arbitrary length $T$ and \emph{randomly} outputs a \emph{transcript} $o$ in some measurable set $O\subseteq\mc{O}$ where $\mc{O}=\textrm{range}(\mc{M})$ is the output universe. DP on data streams is defined for the randomized mechanism $\mc{M}$ as follows. 
	
	\begin{defn}\label{defn:dp}
		A randomized mechanism $\mc{M}$ operating on data streams satisfies event-level (user-level) $\epsilon$-differential privacy, if for all measurable sets $O\subseteq\mc{O}$, all event-level (resp. user-level) adjacent stream prefixes $S^T, \tilde{S}^T$ and all $T$, it holds that,
		\begin{align}
			P[\mc{M}(S^T)=o] \leq e^\epsilon\cdot P[\mc{M}(\tilde{S}^T)=o],\ \forall o\in O.
		\end{align}
	\end{defn}
	
	Differential privacy on data streams ensures that the distribution of the output $o$ reveals limited information about the input $S^T$: for any other event or user-level adjacent input $\tilde{S}^T$, the output under $\tilde{S}^T$ has a similar
	distribution to that under $S^T$. The maximum distance between output distributions is bounded by $\epsilon$ in log scale, therefore $\epsilon$-DP provides a strong worst-case guarantee. It also controls the tradeoff between the false-alarm (Type I) and missed-detection (Type II) errors for an adversary trying to make a hypothesis test between $S^T$ and $\tilde{S}^T$~\cite{kairouz2015composition}. Smaller $\epsilon$ means greater indistinguishability between output distributions given adjacent inputs and hence less privacy risk. We achieve perfect privacy when $\epsilon=0$, and essentially the output becomes independent from the input.
	
	\begin{figure}[t]
		\centering
		\begin{subfigure}[]{0.055\linewidth}
			\centering
			\includegraphics[width=\linewidth,height=4\linewidth]{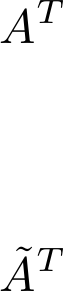}
		\end{subfigure}
		\begin{subfigure}{0.4\linewidth}
			\centering
			\includegraphics[width=.6\linewidth,height=.9\linewidth]{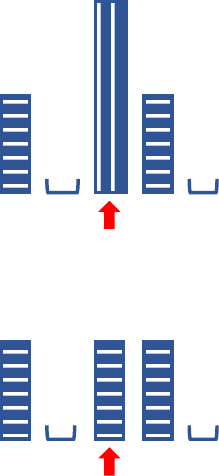}
			\caption{Packet-size adjacent.}
			\label{fig:adjacency_type}
		\end{subfigure}\hspace{.02\linewidth}
		\begin{subfigure}{0.4\linewidth}
			\centering
			\includegraphics[width=.6\linewidth,height=.9\linewidth]{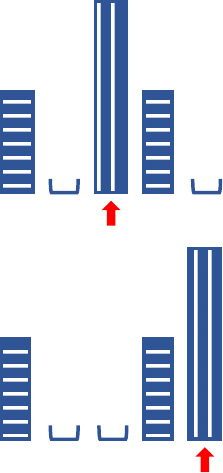}
			\caption{Packet-timing adjacent.}
			\label{fig:adjacency_timing}
		\end{subfigure}
		\caption{Illustration for different pairs of event-level adjacent packet stream prefixes $A^T$ and $\tilde{A}^T$ for $T=5$.}
		\label{fig:adjacency}
	\end{figure}
	
	\subsection{Privacy Model on Event Packet Streams} \label{sec:dp_on_pkt_stream}
	
	In the IoT setting, we are interested in scenarios where the data set $X$ is updated by the event packet stream $\{A_t\}$. For this work, we focus on the event-level privacy model. As variable packet sizes and timing are different kinds of information, we further define separate notions of event-level adjacency for packet stream prefixes $A^T$ and $\tilde{A}^T$. 
	
	\begin{defn}\label{defn:adjacency}
		Two packet stream prefixes $A^T$ and $\tilde{A}^T$ are 
		\begin{enumerate}
			\item \textbf{event-level packet-size} adjacent, if there is at most one time slot $t$ in which $0<A_t\neq \tilde{A}_t$. 
			\item \textbf{event-level packet-timing} adjacent, if there is at most one event (non-zero) packet size appearing in both $A^T$ and $\tilde{A}^T$ but in 2 different time slots $t\neq s$ respectively, that $A_t=\tilde{A}_s>0$ and $\tilde{A}_t=A_s=0$.
		\end{enumerate}
	\end{defn} %
	
	The adjacency definition captures the type and timing information an event-level adversary is interested in. We visualize 2 event-level adjacent pairs of $A^T$ and $\tilde{A}^T$ in Fig.~\ref{fig:adjacency}, with Fig.~\ref{fig:adjacency_type} showing the difference in the event packet sizes, hence the type of an event (e.g., motion detection vs checking camera feed) and Fig.~\ref{fig:adjacency_timing} displaying the different timing of an event (e.g., user going to bed at 10pm vs 11pm). In order to obfuscate such differences from the adversary, we want the same random output $D^T$ by the shaping mechanism $\mc{M}:\mc{A}^T\rightarrow\mc{D}^T$ to have similar probabilities coming from either $A^T$ or $\tilde{A}^T$. Applying Definition~\ref{defn:dp} to the event-level adjacent packet stream prefixes, we establish the privacy model for the randomized shaping mechanism $\mc{M}$ as follows.
	\begin{defn} \label{defn:dp_ind}
		Let $T$ be given and let $\mc{M}: \mc{A}^T\rightarrow\mc{D}^T$ denote a randomized shaping mechanism that takes as input a length-$T$ prefix of the event packet stream $A^T$ and outputs a same-length packet stream $D^T$. Then $\mc{M}$ is event-level packet-size/packet-timing $\epsilon$-differentially private if for all measurable sets $O\subseteq \mc{D}^T$, all event-level packet-size/packet-timing adjacent stream prefixes $A^T, \tilde{A}^T$ and all $T$, it holds that,
		\begin{align}
			\epsilon_\textrm{DP}(\mc{M})\triangleq\left|\log\frac{P[\mc{M}(A^T)=\mbf{d}^T]}{P[\mc{M}(\tilde{A}^T)=\mbf{d}^T]}\right| \leq \epsilon,\ \forall \mbf{d}^T\in O. \label{eqn:dp_defn}
		\end{align}
	\end{defn}
	
	Additionally, we use the following definition to abbreviate the privacy guarantees of the shaping mechanism $\mc{M}$.
	\begin{defn} \label{defn:dp_all}
		A randomized shaping mechanism $\mc{M}$ is event-level $(\epsilon_s, \epsilon_t)$-DP if it is event-level packet-size $\epsilon_s$-DP and packet-timing $\epsilon_t$-DP according to Definition~\ref{defn:dp_ind}.
	\end{defn}
	
	We say that a shaping mechanism is event-level private if it offers both event-level packet-size and packet-timing privacy. Note that however, an event-level private shaping mechanism cannot hide the presence or absence of a user. Heavy traffic may be generated in the presence of an active user and nothing is generated when the user is absent, creating packet stream prefixes that are user-level adjacent according to Section~\ref{section:stream_privacy}. Guaranteeing user-level privacy then requires the shaping mechanism to send a lot of dummy traffic even when the user is absent~\cite{apthorpe2019keeping}. This is generally too costly in terms of shaping overhead, and the effort may easily be in vain when the user's location can be inferred from other sources. The event-level privacy model over infinite packet streams are therefore applicable and practical when faced with an event-level adversary. 
	
	\subsection{Memoryless Shaper}
	
	There are two classes of shaping mechanisms: with \emph{memory} and \emph{memoryless}. Designing an overhead-efficient randomized shaper with memory involves finding the overhead-minimal stochastic mapping from all possible $A^T\in\mc{A}^T$ to $D^T\in\mc{D}^T$. This requires performing very high dimensional thus computationally expensive optimization for large $T$ and quickly becomes intractable as $T\rightarrow\infty$. Nonetheless, the saving on the shaping overhead is only marginal~\cite{iacovazzi2014internet}.
	
	In this paper, we focus on designing event-level DP shapers that are memoryless: $P(D_t|A^t,D^{t-1})=P(D_t|A_t)$. This simplifies the construction, optimization and analysis of the shaper with significantly reduced search space. Moreover, the current departure of the memoryless shaping mechanism only leaks information about the current arrival, whereas shaping mechanisms with memory leak additional information about all the past arrivals through the current output. The memoryless shaper can then be treated as a sequence $\{\mc{M}_t\}$ of independent mechanisms $\mc{M}_t:A_t\rightarrow D_t$, and we have,
	\begin{align}
		\epsilon_\textrm{DP}(\mc{M}) = \left|\sum_{t=1}^T\log\frac{P[\mc{M}_t(A_t)=d_t]}{P[\mc{M}_t(\tilde{A}_t)=d_t]}\right|. \label{eqn:dp_memoryless}
	\end{align}
	
	Next, we will show how to design a memoryless shaping mechanism that satisfies the event-level $(\epsilon_s,\epsilon_t)$-DP guarantee to mask the packet sizes and timing information.
	
	\section{Traffic Shaping Mechanisms} \label{sec:shaping_mechanisms}
	We start by proposing a memoryless shaping mechanism called \textbf{DPS} (\textbf{d}ifferentially-\textbf{p}rivate \textbf{s}haper). We denote it by $\mc{M}^\textrm{DPS}$ and prove its event-level $(\epsilon_s, \epsilon_t)$-DP guarantee according to Definition~\ref{defn:dp_all}. Furthermore, we show that under our IoT traffic model, different settings of $(\epsilon_s, \epsilon_t)$ let the mechanism recover some existing shapers proposed in the literature. Specifically, we name 2 other shaping mechanisms \textbf{PST} (\textbf{p}erfect privacy for both packet \textbf{s}izes and \textbf{t}iming, by setting $\epsilon_s$, $\epsilon_t=0$) and \textbf{PPS} (\textbf{p}erfect privacy for \textbf{p}acket \textbf{s}izes only, by setting $\epsilon_s=0$, $\epsilon_t=\infty$) and denote them by $\mc{M}^\textrm{PST}$ and $\mc{M}^\textrm{PPS}$, respectively. We will also highlight the differences in their privacy guarantees and overhead measures. 
	

	\subsection{Differentially Private Shaping Mechanism} \label{sec:DPS}
	
	The DPS mechanism is a sequence $\{\mc{M}^\textrm{DPS}_t\}$ of independent mechanisms $\mc{M}^\textrm{DPS}_t:\mc{A}\xrightarrow{c}\mc{D}$ acting as a discrete memoryless channel (DMC) $c$ which defines a conditional probability distribution $c_{\mc{D}|\mc{A}}(d|a)$ with $c_{ij}\triangleq c(d_j|a_i), \forall i\in N, j\in M$. The channel matrix $\mbs{C}=[c_{ij}]_{i\in N, j\in M}$ is right stochastic. Upon a packet arrival of size $A_t=a_i$, the shaper randomly issues a departure packet of size $D_t=d_j$ with probability (w.p.) $c_{ij}$. That is,
	\begin{align}
		\mc{M}^{\textrm{DPS}}_t(A_t) = d_{j}\in \mc{D}\quad \textrm{w.p.}\ c_{ij}\quad \textrm{if } A_t=a_i. \label{eqn:DPS}
	\end{align}
	
	By controlling the dependency between $A_t$ and $D_t$ via the channel $c$, the DPS mechanism can make different choices of $D^T$ given $A^T$ to control the privacy leakage and shaping overhead. To guarantee DP for event packet streams following Definition~\ref{defn:dp_all}, we let the channel $c$ satisfy a more stringent privacy model -- \emph{local differential privacy} (LDP)~\cite{warner1965randomized,kalantari2016optimal}. 
	
	\begin{defn} \label{defn:LDP}
		A channel $c:\mc{A}\rightarrow\mc{D}$ satisfies $\epsilon$-LDP~\emph{\cite{warner1965randomized}} if ~$\max \left\{\frac{c(d|a)}{c(d|\tilde{a})}\right\} \leq e^\epsilon, \forall(a,\tilde{a},d)\in\mc{A}^2\times\mc{D}$.
	\end{defn}
	The goal of designing an $\epsilon$-LDP channel $c$ is to ensure that the adversary's likelihood of guessing that the packet size is $a\in\mc{A}$ over $\tilde{a}\in\mc{A}$ does not increase, multiplicatively, more than $e^\epsilon$ after seeing the obfuscated packet size $d\in\mc{D}$. Therefore, in each time slot, the adversary is limited in inferring about the actual user activity instance from the eavesdropped packet size. The shaper $\mc{M}^\textrm{DPS}_t$ adopting an LDP memoryless channel $c$ has the following privacy guarantee.
	
	\begin{prop}\label{prop:DPS}
		The shaping mechanism $\mc{M}^\textrm{DPS}_t:\mc{A}\xrightarrow{c}\mc{D}$ with a $\max\left(\epsilon_s, \frac{\epsilon_t}{2}\right)$-LDP memoryless channel $c$ satisfying the following constraints,
		\begin{align}
			\left\{
			\begin{array}{ll}
				c_{ij} - e^{\epsilon_s}\cdot c_{kj}\leq 0,\ &\forall i,k\in N^+; \forall j\in M, \\
				c_{ij} - e^{\epsilon_t/2}\cdot c_{0j}\leq 0,\ &\forall i\in N^+; \forall j\in M, \\
				c_{0j} - e^{\epsilon_t/2}\cdot c_{kj}\leq 0,\ &\forall k\in N^+; \forall j\in M. 
			\end{array}\right. \label{eqn:priv_constr_alt}
		\end{align}
		guarantees event-level $(\epsilon_s, \epsilon_t)$-DP as stated in Definition~\ref{defn:dp_all}.
	\end{prop}
	\begin{proof}
		See Appendix~\ref{appdix:DPS}.
	\end{proof}
	
	The intuition behind the set of constraints is that, i) given 2 input packet sizes $a_i\neq a_k$, the channel $c$ will output the same packet size $d_j$ with similar probabilities, whose ratio is bounded by $e^{\epsilon_s}$ or $e^{\epsilon_t/2}$; ii) for protecting timing information, the privacy budget $\epsilon_t$ is equally divided into the 2 time slots where event-level packet-timing adjacent prefixes $A^T$ and $\tilde{A}^T$ differ, as shown in Fig~\ref{fig:adjacency_timing}.
	
	The output byte rate and transmission efficiency level of the DPS mechanism given $c$ are,
	\begin{align}
		B^{\textrm{DPS}}_{out} 
		& = \mbb{E}_{\lambda\times c}[D_t] 
		= \sum_{i\in N}\lambda_i\sum_{j\in M}c_{ij}d_j
		= \mbs{\lambda}^\top\mbs{C}\mbs{d},
		\label{eqn:out_byte_rate_dps} \\
		\rho^\textrm{DPS}
		& = B_{in}/B^\textrm{DPS}_{out}
		= \mbs{\lambda}^\top\mbs{a}/\mbs{\lambda}^\top\mbs{C}\mbs{d}. 
		\label{eqn:rho_dps}
	\end{align}
	
	The complete tunability of privacy levels $(\epsilon_s,\epsilon_t)$ in protecting the packet sizes and timing information makes the DPS mechanism well adaptable to heterogeneous network conditions and user demands for privacy. With this, a traffic shaping system can control the privacy parameters to be able to interpolate between a system that guarantees perfect privacy and one without any privacy protection. We can easily verify that setting $\epsilon_s, \epsilon_t=\infty$ in the constraints allows $\mbs{C}$ to become an identity matrix (assuming $\mc{A}=\mc{D}$ w.l.o.g.). The shaper with an identity channel matrix outputs $D_t=A_t, \forall t$ and offers no privacy protection. In the sequel, we look at two other non-trivial settings in the perfect privacy regime.
	

	\subsection{Perfect Privacy Shaping Mechanism} \label{sec:PST}
	
	The PST mechanism is a special case of the DPS mechanism by setting $\epsilon_s, \epsilon_t=0$ in the constraints~\eqref{eqn:priv_constr_alt}. By simple algebra, this setting forces $c_{ij}=c_{kj}, \forall i,k\in N; \forall j\in M$. Then the channel matrix becomes rank-one:  $\mbs{C}=\mbf{1}\cdot\mbs{\mu}^\top$, with all rows equal to the same probability vector $\mbs{\mu}^\top=[\mu_0,\mu_1,\ldots,\mu_m]$, which defines a probability distribution $\mu(d)$ on $\mc{D}$ with $\mu_j\triangleq\mu(d_j), \forall j\in M$. We can describe PST as,
	\begin{align}
		\mc{M}^{\textrm{PST}}_t(A_t) = d_{j}\in \mc{D}\quad \textrm{w.p.}\ \mu_j\quad \forall A_t\in\mc{A}. \label{eqn:PST}
	\end{align}
	Essentially, the PST mechanism chooses departures $D_t\sim\mbs{\mu}$ \emph{independently} from the arrivals $A_t$ in every time slot and generates a \emph{standardized} packet stream that won't leak any information about private user activities.
	\begin{prop} \label{prop:PST}
		$\mc{M}^{\textrm{PST}}$ guarantees perfect event-level privacy, or $(0,0)$-DP according to Definition~\ref{defn:dp_all}.
	\end{prop}
	\begin{proof}
		See Appendix~\ref{appdix:PST}.
	\end{proof}
	
	The kind of standardization performed by PST is similarly introduced as the ``fixed pattern masking'' by Iacovazzi and Baiocchi~\cite{iacovazzi2014internet}. The output byte rate and transmission efficiency level of the PST mechanism given $\mu$ are,
	\begin{align}
		B^\textrm{PST}_{out}
		& = \mbb{E}_\mu[D_t]=\sum_{j\in M}\mu_jd_j=\mbs{\mu}^\top\mbs{d}.
		\label{eqn:byte_rate_pst}\\
		\rho^\textrm{PST}
		& = B_{in}/B^\textrm{PST}_{out}=\mbs{\lambda}^\top\mbs{a} / \mbs{\mu}^\top\mbs{d}, \label{eqn:efficiency_pst}
	\end{align}
	
	\begin{table*}[t]
		\centering
		\small
		\begin{tabular}[c]{|c|>{\centering\arraybackslash}m{4.5cm}|>{\centering\arraybackslash}m{4.5cm}|>{\centering\arraybackslash}m{5.5cm}|}
			\hline
			Shapers & DPS 
			& PST 
			& PPS \\
			\hline
			Description &   \begin{eq}
				\mc{M}^{\textrm{DPS}}_t(A_t) = d_{j}\ \textrm{w.p.}\ c_{ij}\ \textrm{if } A_t=a_i
			\end{eq}
			&   \begin{eq}
				\mc{M}^{\textrm{PST}}_t(A_t) = d_{j}\ \textrm{w.p.}\ \mu_j\ \forall A_t\in\mc{A}
			\end{eq}
			&   \begin{eq}
				\mc{M}^{\textrm{PPS}}_t(A_t) 
				= \left\{
				\begin{array}{lll}
					0 & \textrm{w.p.}\ 1 & \text{if } A_t = 0\\
					d_{j} & \textrm{w.p.}\ \upsilon_j & \text{if } A_t > 0
				\end{array} \right.
			\end{eq} \\
			\hline
			Privacy & $(\epsilon_s, \epsilon_t)$-DP  
			& Perfect event-level privacy  
			& Perfect event-level privacy for packet sizes \\ 
			\hline
			Efficiency 
			& $\rho^\textrm{DPS}=\mbs{\lambda}^\top\mbs{a}/\mbs{\lambda}^\top\mbs{C}\mbs{d}$         
			& $\rho^\textrm{PST}=\mbs{\lambda}^\top\mbs{a}/\mbs{\mu}^\top\mbs{d}$         
			& $\rho^\textrm{PPS}=\mbs{\lambda}^\top\mbs{a}/\mbs{\upsilon}^\top\mbs{d}\cdot\frac{1}{\Lambda}$ \\ 
			\hline
			Optimization
			& \eqref{opt:DPS}
			&
			\begin{eq}
				\begin{aligned}
					\min_{\mbs{\mu}}\ & \mbb{E}[Q_t]\\
					\subjto\ 
					& \mbs{\mu}^\top\mbs{d}\cdot \rho \leq \mbs{\lambda}^\top\mbs{a}, \\
					& \mbs{\mu}^\top\mbs{d} > \mbs{\lambda}^\top\mbs{a}, \\
					& 0 \preceq \mbs{\mu} \preceq 1,\\
					& \mbf{1}^\top\mbs{\mu} = 1.
				\end{aligned}
				\tag{$\mc{P}^\textrm{PST}$} \label{opt:PST}
			\end{eq}
			&
			\begin{eq}
				\begin{aligned} 
					\min_{\mbs{\upsilon}}\ & \mbb{E}[Q_t]\\
					\subjto\ 
					& \Lambda\mbs{\upsilon}^\top\mbs{d}\cdot \rho \leq \mbs{\lambda}^\top\mbs{a}, \\
					& \Lambda\mbs{\upsilon}^\top\mbs{d} > \mbs{\lambda}^\top\mbs{a}, \\
					& 0 \preceq \mbs{\upsilon} \preceq 1,\\
					& \mbf{1}^\top\mbs{\upsilon} = 1.
				\end{aligned}
				\tag{$\mc{P}^\textrm{PPS}$} \label{opt:PPS}
			\end{eq} \\
			\hline
		\end{tabular}
		\caption{Comparison between shaping mechanisms in the design and optimization, privacy guarantees and overhead measures.}
		\label{tab:shapers_cmp}
	\end{table*}
	
	\subsection{Perfect Privacy for Packet Sizes Only}\label{sec:PPS}
	
	In many cases, network conditions or user demands for privacy may place constraints on the delay and byte rate overhead which preclude perfect event-level privacy for both packet sizes and timing. In these scenarios, an alternative kind of shaping mechanisms relax the privacy guarantee to focus on perfect event-level privacy for packet sizes only. 
	
	Here, we describe such a shaping mechanism PPS: it standardizes only the sizes in the event packet stream without changing the timing information. That is, PPS mechanism will output packets only when event packets $A_{t\in\mc{I}}$ arrive at the switch. Then, the output packet sizes $D_{t\in\mc{I}}$ are drawn i.i.d. from a chosen distribution $\upsilon(d)$ on $\mc{D}$ similarly defined as $\mu(d)$ in the PST mechanism, with $\mbs{\upsilon}=[\upsilon_0,\ldots,\upsilon_m]^\top$ denoting the probability vector. If there is no arrival ($A_t=0$), nothing is sent out ($D_t=0$). We describe the PPS mechanism by, 
	\begin{align}
		\mc{M}^{\textrm{PPS}}_t(A_t) = \left\{
		\begin{array}{lll}
			0 & \textrm{w.p.}\ 1 & \text{if } A_t = 0,\\
			d_{j}\in \mc{D} & \textrm{w.p.}\ \upsilon_j & \text{if } A_t > 0.
		\end{array} \right. \label{eqn:PPS}
	\end{align}
	
	The PPS mechanism is another special case of DPS with $\epsilon_s=0, \epsilon_t=\infty$. One can easily show that it corresponds to a channel matrix $\mbs{C}$ with the first row as $[1, 0, \ldots, 0]$ and the rest of the rows all equal to $\mbs{\upsilon}^\top$. While preventing privacy leakage via packet sizes, PPS permits the timing information to remain unchanged and is equivalent to Wright's traffic morphing method~\cite{wright2009traffic} when packet fragmentation is allowed.
	
	\begin{prop} \label{prop:PPS}
		$\mc{M}^{\textrm{PPS}}$ guarantees perfect event-level packet-size privacy, or $(0,\infty)$-DP according to Definition~\ref{defn:dp_all}.
	\end{prop}
	\begin{proof}
		See Appendix~\ref{appdix:PPS}.
	\end{proof}
	
	The output byte rate and transmission efficiency level of the PPS mechanism given $\upsilon$ are,
	\begin{align}
		B^{\textrm{PPS}}_{out}
		& = \mbb{E}[D_t]
		= \lambda_0\mbb{E}[D_t|A_t=a_0]+\Lambda\mbb{E}_{\upsilon}[D_t|A_t\neq a_0]\nonumber\\
		& = \lambda_0d_0+\Lambda\sum_{j\in M}\upsilon_jd_j
		= \Lambda\mbs{\upsilon}^\top\mbs{d}. 
		\label{eqn:byte_rate_pps}\\
		\rho^\textrm{PPS}
		& = B_{in}/B^\textrm{PPS}_{out}
		= \mbs{\lambda}^\top\mbs{a}/\mbs{\upsilon}^\top\mbs{d}\cdot\frac{1}{\Lambda}. \label{eqn:efficiency_pps}
	\end{align}
	
	Comparing~\eqref{eqn:byte_rate_pps} with~\eqref{eqn:byte_rate_pst} and~\eqref{eqn:efficiency_pps} with~\eqref{eqn:efficiency_pst}, we see that for $\Lambda\in(0,1)$~\eqref{eqn:arrival_rate} and if $\mbs{\mu}=\mbs{\upsilon}$, $B^\textrm{PPS}_{out}=\Lambda B^\textrm{PST}_{out}<B^\textrm{PST}_{out}$ and $\rho^\textrm{PPS}=\rho^\textrm{PST}/\Lambda>\rho^\textrm{PST}$. In general, a traffic shaping mechanism can be more transmission efficient for satisfying a less restrictive privacy guarantee.
	
	\section{Delay-Optimal Shapers} \label{sec:OPT}
	
	In order to efficiently utilize the limited resources in IoT networks, it is crucial to optimize the DPS, PST and PPS mechanisms to incur minimal shaping overhead. In this section, we formulate the problem of finding overhead-optimal shapers as constrained optimizations: for a given transmission efficiency level $\rho\in(0,1)$ that ensures queue stability ($\bar{W}_p\rightarrow\mbb{E}[W_p]<\infty$ and $\bar{Q}_t\rightarrow\mbb{E}[Q_t]<\infty$), we wish to find optimal $\mbs{C}^*$, $\mbs{\mu}^*$ and $\mbs{\upsilon}^*$ that minimize the delay overhead $\mbb{E}[W_p]$.
	
	Expressing and analyzing $\mbb{E}[W_p]$ as a function of $\mbs{C}^*$, $\mbs{\mu}^*$ or $\mbs{\upsilon}^*$ is however intractable. Because of packet splitting and padding, the waiting time $W_p$ of event packets depends heavily on both the past and future arrivals and departures. For example, using Fig.~\ref{fig:shaper_diagram} again for illustration, and consider the constant departures as random draws from $\mbs{\mu}$ by the PST mechanism. The arrival $A_1=a_i$ waits $W_1=1$ time slot because: i) the queue was depleted before $A_1$, ii) $D_1=d<A_1$ is drawn with probability $\lambda_i\cdot\sum_{d<a_i,d\in\mc{D}}\mu(d)$, iii) $A_2=a_j$ arrives with probability $\lambda_j$, and iv) $D_2=d>(A_1+A_2-D_1)=(a_i+a_j-d)$ is drawn with probability $\lambda_i\lambda_j\cdot\sum_{d>(a_i+a_j)/2,d\in\mc{D}}\mu(d)$. Due to such complicated dependency, it is beyond the reach to establish a clear relationship between $\mbb{E}[W_p]$ and $\mbs{\mu}$.
	
	Conversely, the discrete-time Lindley's equation~\eqref{eqn:lindley_qsz} describing the evolution of queue size in terms of arrival and departure pairs $(A_t,D_t)$ provides a straightforward mathematical model that allows for subsequent analysis. We rewrite it as,  
	\begin{align}
		Q_t = \max(Q_{t-1}+X_t,0), \label{eqn:lindley_qsz2}
	\end{align}
	where $X_t\triangleq A_t-D_t$ is i.i.d. across time by the i.i.d. assumption on $A_t$ and the memoryless property of the mechanisms. For $\rho=B_{in}/B_{out}\in(0,1)$, we have $\mbb{E}[X_t]=\mbb{E}[A_t]-\mbb{E}[D_t]=(\rho-1)B_{out}<0$. According to~\cite{grassmann1989numerical}, the expected queue size $\mbb{E}[Q_t]$ can then be very efficiently and well approximated with the use of Wiener-Hopf factorization (WHF) method.
	
	Based on Little's law~\cite{little2008little}, the smaller the queue size is, the less waiting time the event packets will experience. We then propose minimizing the expected queue size across time $\mbb{E}[Q_t]$ as a proxy for the minimization of the expected delay per event packet $\mbb{E}[W_p]$\footnote{Note that by Little's law, $\mbb{E}[Q_t]=\hat{\Lambda}\mbb{E}[W_t]$ where $\hat{\Lambda}$ and $\mbb{E}[W_t]$ are the arrival rate and expected waiting time at the byte level, different from those at the packet level ($\Lambda$ and $\mbb{E}[W_p]$). A shaper that minimizes $\mbb{E}[Q_t]$ simultaneously minimizes $\mbb{E}[W_t]$. It may not yield the exact minimum of $\mbb{E}[W_p]$ but should be good enough in practice. We therefore propose minimizing $\mbb{E}[Q_t]$ as proxy for the minimization of $\mbb{E}[W_p]$ thanks to the tremendous computational advantage of WHF.} to find the delay-optimal shapers. We will also justify this choice by empirical results in Section~\ref{sec:EXP}. More importantly, we will show that $\mbb{E}[Q_t]$ enjoys a nice property of being convex in the optimization variables $\mbs{C}$, $\mbs{\mu}$ or $\mbs{\upsilon}$.
	
	\subsection{Delay-Optimal DPS} \label{sec:delay_opt_dps}
	
	For given values of $(\epsilon_s, \epsilon_t)$ and $\rho\in(0,1)$ and a given set of output packet sizes $\mc{D}$, we can find the optimal $\max\left(\epsilon_s, \frac{\epsilon_t}{2}\right)$-LDP channel matrix $\mbs{C}^*$ that minimizes the expected queue size $\mbb{E}[Q_t]$ by solving the following optimization problem,
	\begin{equation}
		\begin{aligned} 
			\min_{\mbs{C}}\quad & \mbb{E}[Q_t] \\
			\subjto
			\quad & \mbs{\lambda}^\top\mbs{C}\mbs{d}\cdot\rho \leq \mbs{\lambda}^\top\mbs{a}, \\
			& \mbs{\lambda}^\top\mbs{C}\mbs{d} > \mbs{\lambda}^\top\mbs{a}, \\
			& \mbs{C}\ \textrm{satisfies~\eqref{eqn:priv_constr_alt}}, \\
			& 0\preceq \mbs{C} \preceq 1, \\
			& \mbs{C}\cdot\mbs{1}_{m+1} = \mbs{1}_{n+1}.
		\end{aligned}
		\tag{$\mc{P}^\textrm{DPS}$}\label{opt:DPS}
	\end{equation}
	The first constraint makes sure that the transmission efficiency $\rho^\textrm{DPS}$~\eqref{eqn:rho_dps} is at least $\rho$. The second constraint ensures a stable FCFS queue so that the optimal value is finite. The third constraint enforces a $\max\left(\epsilon_s, \frac{\epsilon_t}{2}\right)$-LDP channel matrix $\mbs{C}$ for the DPS mechanism to satisfy event-level $(\epsilon_s,\epsilon_t)$-DP according to Proposition~\ref{prop:DPS}. The last 2 constraints restrict $\mbs{C}$ to be right stochastic, where $\preceq$ is element-wise comparison.
	
	\begin{algorithm}[t]
		\SetAlgoLined
		\SetKwInput{KwData}{Cache}
		\caption{DPS mechanism}
		\label{alg:DPS}
		\LinesNumbered
		\KwIn{Arrival packet $p_a$ at time $t$; $\mbs{C}^*$}
		\KwOut{Departure packet $p_d$ at time $t$}
		\KwData{FCFS queue $Q$ with $q$ bytes at time $t$}
		\eIf{(arrival packet $p_a$)}{
			$a \gets$ size($p_a$); \\
			Enqueue(Q, $p_a$); \\
			$q \gets q + a$; 
		}{
			$a \gets 0$;
		}
		Generate departure packet size $d \sim c^*_{\mathcal{D}|\mathcal{A}}(d|a)$\;
		\If{$d>0$}{
			\eIf{$q<d$} {
				$tmp \gets$ Dequeue($Q$, $q$); \\
				$p_d \gets$ Pad($tmp$, $d-q$ dummy bytes); \\
				$q \gets 0$;
			}{
				$p_d \gets$ Dequeue($Q$, $d$); \\
				$q \gets q-d$;
			}
			Push departure packet $p_d$ out.
		}
	\end{algorithm}
	
	\subsubsection{Solution Overview}
	
	To utilize the DPS mechanism, the shaping system first gathers prior information about the traffic statistics of networked IoT devices. It infers about the set of possible packet sizes $\mc{A}$ and their arrival rates $\mbs{\lambda}$ by observing the traffic offline for a period of time, that the devices monitor user activities and generate traffic as they are but the system does not shape the traffic for output just yet. The system can discretize the traffic by the minimum packet interval to ensure one packet per time slot and calculate $\mbs{\lambda}$ as the empirical probability mass function (PMF) on $\mc{A}$. 
	
	Once a user specifies the transmission efficiency $\rho\in(0,1)$ and privacy preferences $(\epsilon_s,\epsilon_t)$, the system solves~\ref{opt:DPS} to find the optimal channel matrix $\mbs{C}^*$, or $c^*_{\mc{D}|\mc{A}}(d|a)$, where $\mc{D}$ can be chosen to be the same as $\mc{A}$. Now the system shapes the input traffic through a FCFS queue following the example in Fig.~\ref{fig:shaper_diagram}, whose input-output relationship is governed by $\mbs{C}^*$. That is, in each time slot $t$, it samples an output size $D_t$ given arrival size $A_t$ according to~\eqref{eqn:DPS}. Algorithm~\ref{alg:DPS} describes how the DPS mechanism processes its input and output for shaping packetized traffic during each time slot $t$.
	
	Every time the network of IoT devices changes or the user specifies a different set of privacy and efficiency parameters, the system updates $\mc{A}$, $\mbs{\lambda}$ and solves for $\mbs{C}^*$ again for shaping.
	
	\subsection{Delay-Optimal PST and PPS}
	We argued in Section~\ref{sec:shaping_mechanisms} that PST ($\mbs{\mu}$) and PPS ($\mbs{\upsilon}$) shapers are special cases of the DPS mechanism. To find the delay-optimal PST (PPS) shaper, we can solve the same optimization problem~\eqref{opt:DPS} for the optimal $\mbs{C}^*$ with $\epsilon_s=\epsilon_t=0$ (\emph{resp.} $\epsilon_s=0$, $\epsilon_t=\infty$) and then infer about the corresponding optimal $\mbs{\mu}^*$ (\emph{resp.} $\mbs{\upsilon}^*$). Alternatively, we can reduce the number of optimization variables by directly constructing optimization problems~\eqref{opt:PST} on $\mbs{\mu}$ and~\eqref{opt:PPS} on $\mbs{\upsilon}$. For simplicity of writing, we summarize the optimization of DPS, PST and PPS mechanisms in Table~\ref{tab:shapers_cmp} along with the differences in their design, privacy guarantees and overhead measures. It turns out that  solving for the optimal shapers are convex programs.
	
	\begin{prop} \label{prop:CVX}
		Given the set of output packet sizes $\mc{D}$ and a transmission efficiency level $\rho\in(0,1)$, the optimization problems~\eqref{opt:DPS},~\eqref{opt:PST} and~\eqref{opt:PPS} are convex programs.
	\end{prop}
	\begin{proof}
		$\mbb{E}[Q_t]$ is convex in $\mbs{C}$, $\mbs{\mu}$ or $\mbs{\upsilon}$ and all constraints in~\eqref{opt:DPS},~\eqref{opt:PST} and~\eqref{opt:PPS} are affine. See Appendix~\ref{appdix:CVX}.
	\end{proof}
	
	\subsection{Privacy-Overhead Tradeoff} \label{sec:POT}
	We can obtain the minimum achievable $\mbb{E}[Q_t]$ for varying privacy parameters $(\epsilon_s, \epsilon_t)$ and transmission efficiency level $\rho$, hence establishing the privacy-overhead tradeoff. The following theorem provides an important qualitative description of the privacy-overhead tradeoff, and its validity will be shown in Section~\ref{sec:EXP} under diverse experimental setups.
	
	\begin{theorem} \label{theorem:POT} 
		Given the set of output packet sizes $\mc{D}$ and transmission efficiency level $\rho\in[\mbs{\lambda}^\top\mbs{a}/d_n, 1)$, the minimum achievable $\mbb{E}[Q_t]$ in~\eqref{opt:DPS} increases when i) $(\epsilon_s,\epsilon_t)$ decreases with $\rho$ fixed and ii) $\rho$ increases with $(\epsilon_s,\epsilon_t)$ fixed.
	\end{theorem}
	\begin{proof}
		This follows from Proposition~\ref{prop:CVX} and standard results based on strong duality and sensitivity analysis~\cite{boyd2004convex}. See Appendix~\ref{appdix:POT} for a detailed proof.
	\end{proof}
	
	\section{Special Policies}
	
	\subsection{Deterministic Policy for PST and PPS} \label{sec:deterministic_policy}
	A folk theorem in queueing theory states that the deterministic service-time distribution with unit mass on a given mean service time minimizes delay when service and inter-arrival times are mutually independent. Humblet~\cite{humblet1982determinism} proved this formally using Jensen's inequality, the convexity of the $\max(x,0)$ function in the Lindley's equation of customer waiting times in a FCFS queue, and the independence between inter-arrival and service times. 
	
	In a similar fashion, we enforce a FCFS queueing discipline for the shapers and describe the evolution of the queue size by the discrete-time Lindley's equation~\eqref{eqn:lindley_qsz} in terms of arrival ($A_t$) and departure ($D_t$) sizes. By analogy, we can make the same statement for the design of shaping mechanisms: when departure sizes $D_t$ are chosen independently from arrival sizes $A_t$, then for a given transmission efficiency level $\rho$ or equivalently an expected output packet size $d^*\triangleq\mbb{E}_\lambda[A_t]/\rho$, the deterministic policy that outputs packets with constant size $D_t=d^*, \forall t$ minimizes the expected queue size $\mbb{E}[Q_t]$. As $A^T$ are i.i.d. across time, PST chooses $D^T$ independently from $A^T$, and PPS selects $D_{t\in\mc{I}}$ independently from $A_{t\in\mc{I}}$\footnote{Lindley's equation~\eqref{eqn:lindley_qsz} applies to the subsequence of queue sizes $Q_{t\in\mc{I}}$ for which $D_{t\in\mc{I}}$ and $A_{t\in\mc{I}}$ are independent and Humblet's result~\cite{humblet1982determinism} still holds.}, enforcing the deterministic policy on PST or PPS should on average accumulate a shorter backlog in the queue than their non-deterministic counterparts.
	
	Let PST* and PPS* be the deterministic versions of PST and PPS mechanisms given $\rho$. That is, PST* generates $D_{t\in\mc{T}}=d^*$ and is in essence the discrete-time version of the traffic shaper by Apthorpe et al.~\cite{apthorpe2017spying} that maintains a constant departure rate in the network traffic leaving a smart home. Likewise, PPS* outputs $D_{t\in\mc{I}}=d^*/\Lambda$. Based on the above statement, if $d^*,d^*/\Lambda\in\mc{D}$, then the optimal solutions to~\eqref{opt:PST} and~\eqref{opt:PPS} are exactly delta distributions: $\mu^*(d)=\delta(d^*)$ and $\upsilon^*(d)=\delta(d^*/\Lambda)$. In real systems, however, $\rho$ may be set arbitrarily by resource-constrained users and the resulting $d^*, d^*/\Lambda$ may not be meaningful values for packet sizes (e.g., not an integer or exceeding the maximum transmission unit).
	
	\begin{figure}[t]
		\centering
		\includegraphics[width=0.725\linewidth]{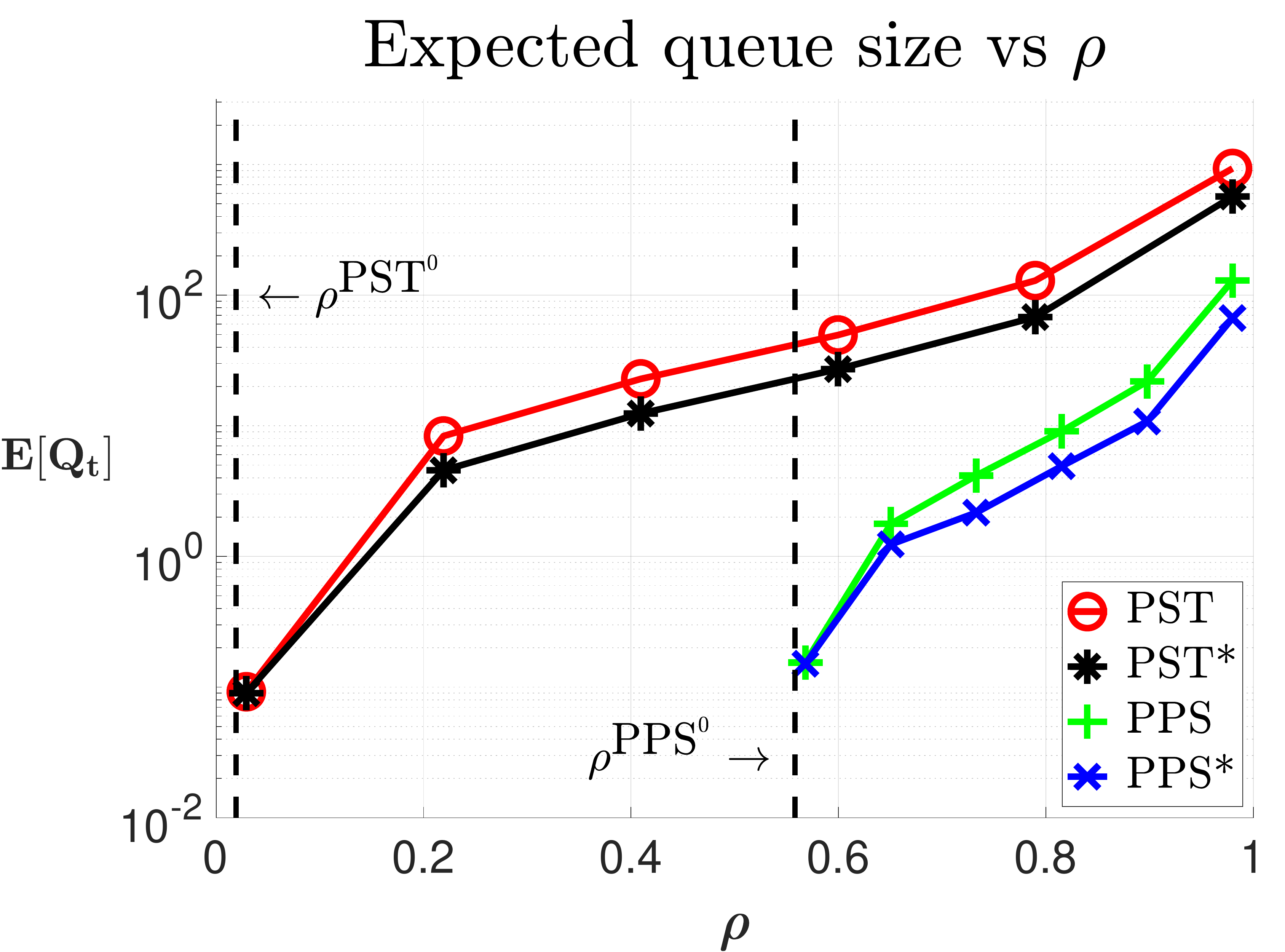}
		\caption{Tradeoffs between expected queue size $\mbb{E}[Q_t]$ and transmission efficiency level $\rho$ for PST/PST* and PPS/PPS* mechanisms. The lowest efficiency levels $\rho^{\textrm{PST}^0}$~\eqref{eqn:efficiency_pst_pad} and $\rho^{\textrm{PPS}^0}$~\eqref{eqn:efficiency_pps_pad} correspond to the PST and PPS shapers with pad-only strategy, both yielding $\mbb{E}[Q_t]=0$. The corresponding data points are not shown in the plot since the y-axis is in log scale.}
		\label{fig:tradeoff_VSCI_VSVI}
	\end{figure}
	
	\def\subfigwd{.31\linewidth}
	\def\subfight{.75\linewidth}
	\begin{figure*}
		\centering
		\begin{subfigure}{\subfigwd}
			\centering
			\includegraphics[width=\linewidth,height=\subfight]{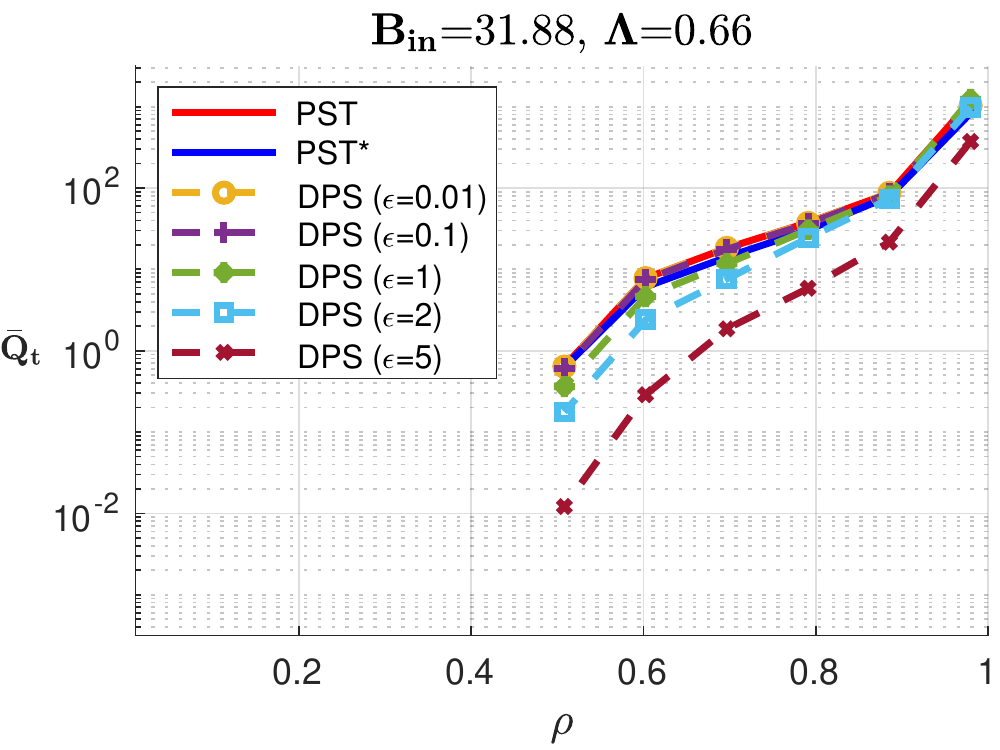}
			\caption{$s = 0.01$.}
			\label{fig:qsz_dps_vs_vsci_s_0_01}
		\end{subfigure} 
		\begin{subfigure}{\subfigwd}
			\centering
			\includegraphics[width=\linewidth,height=\subfight]{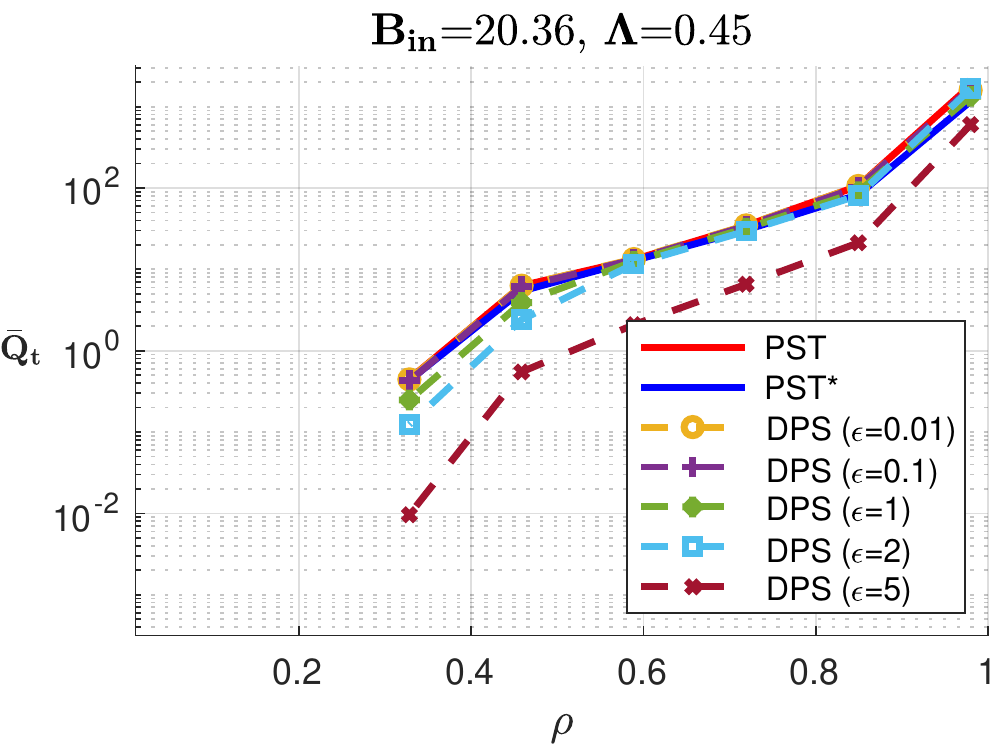}
			\caption{$s = 1$.}
			\label{fig:qsz_dps_vs_vsci_s_1}
		\end{subfigure}
		\begin{subfigure}{\subfigwd}
			\centering
			\includegraphics[width=\linewidth,height=\subfight]{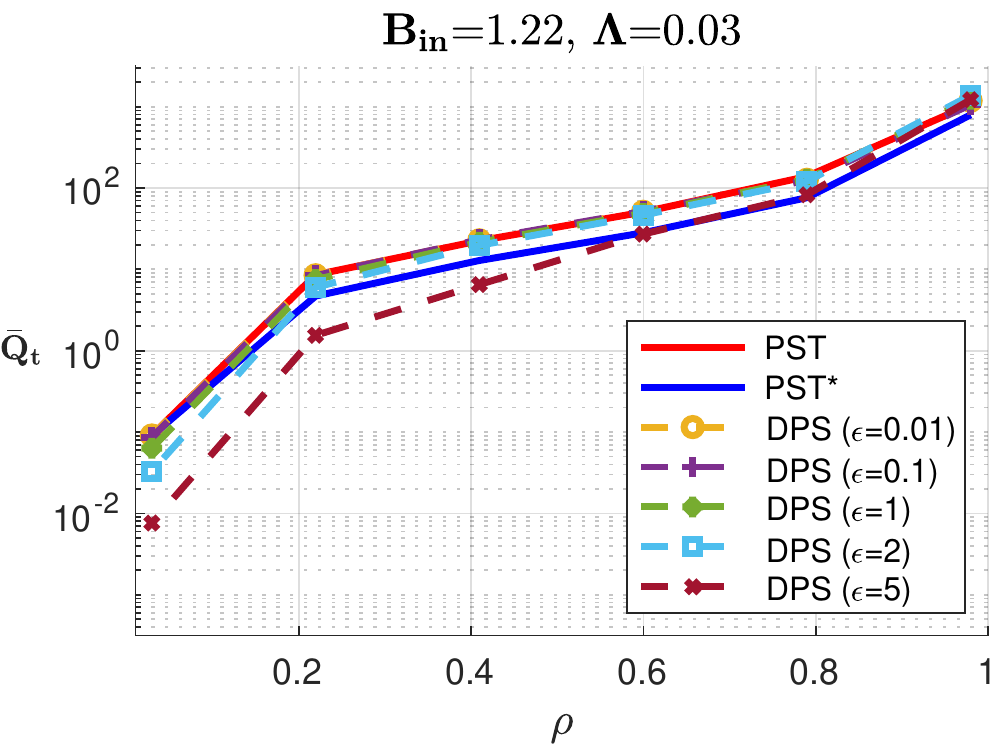}
			\caption{$s = 5$.}
			\label{fig:qsz_dps_vs_vsci_s_5}
		\end{subfigure}\\
		\begin{subfigure}{\subfigwd}
			\centering
			\includegraphics[width=\linewidth,height=\subfight]{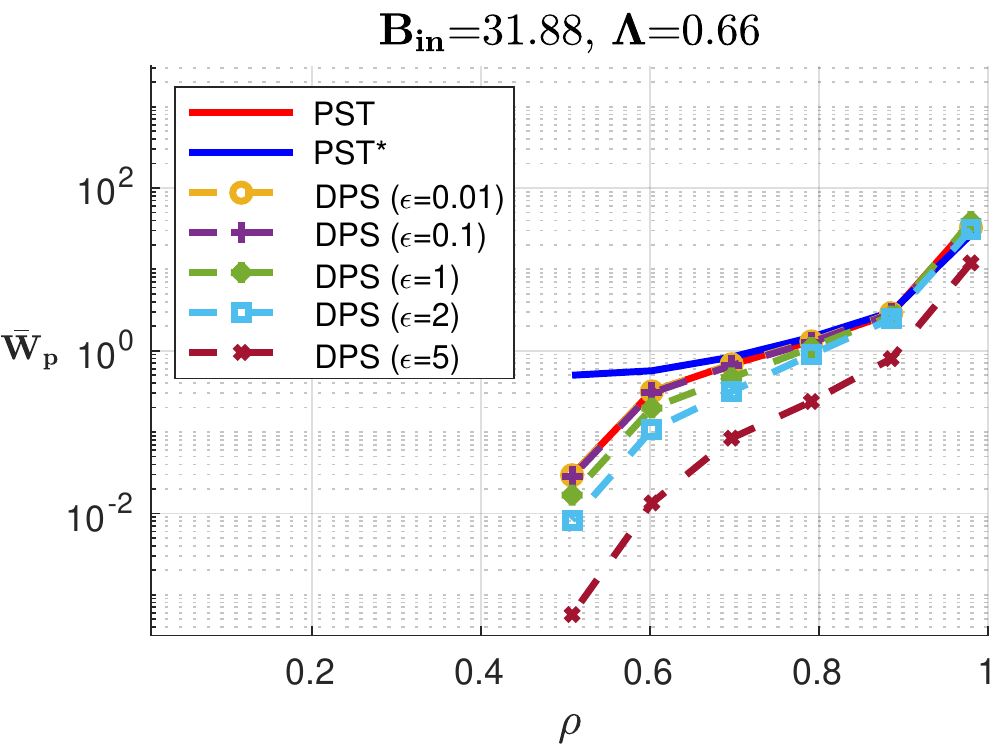}
			\caption{$s = 0.01$.}
			\label{fig:delay_dps_vs_vsci_s_0_01}
		\end{subfigure}
		\begin{subfigure}{\subfigwd}
			\centering
			\includegraphics[width=\linewidth,height=\subfight]{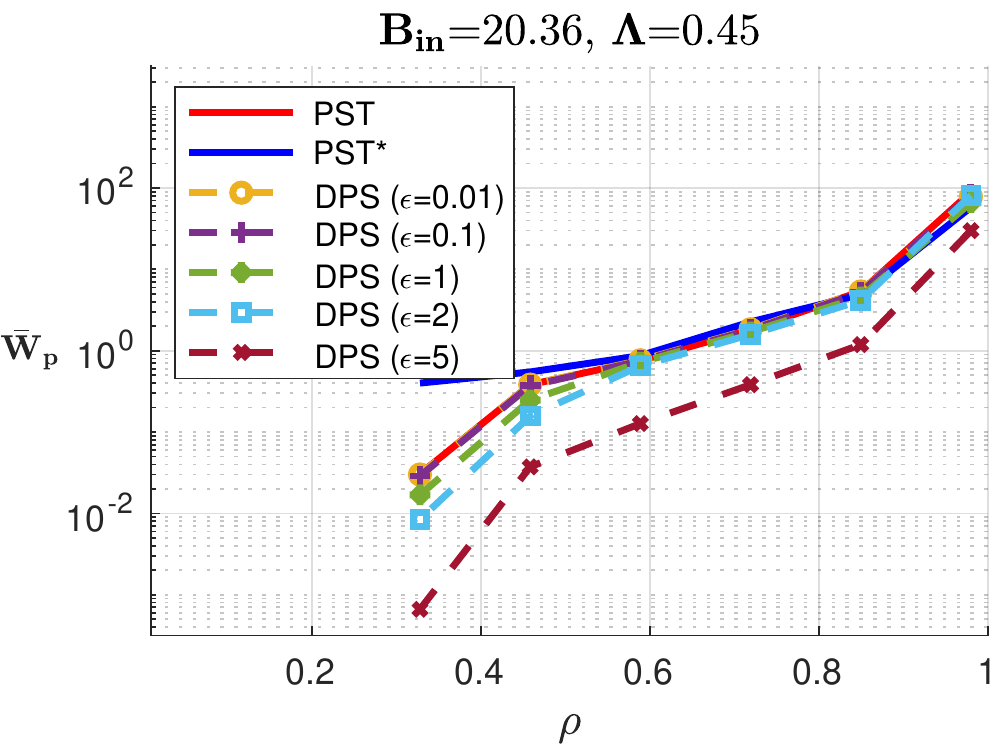}
			\caption{$s = 1$.}
			\label{fig:delay_dps_vs_vsci_s_1}
		\end{subfigure}
		\begin{subfigure}{\subfigwd}
			\centering
			\includegraphics[width=\linewidth,height=\subfight]{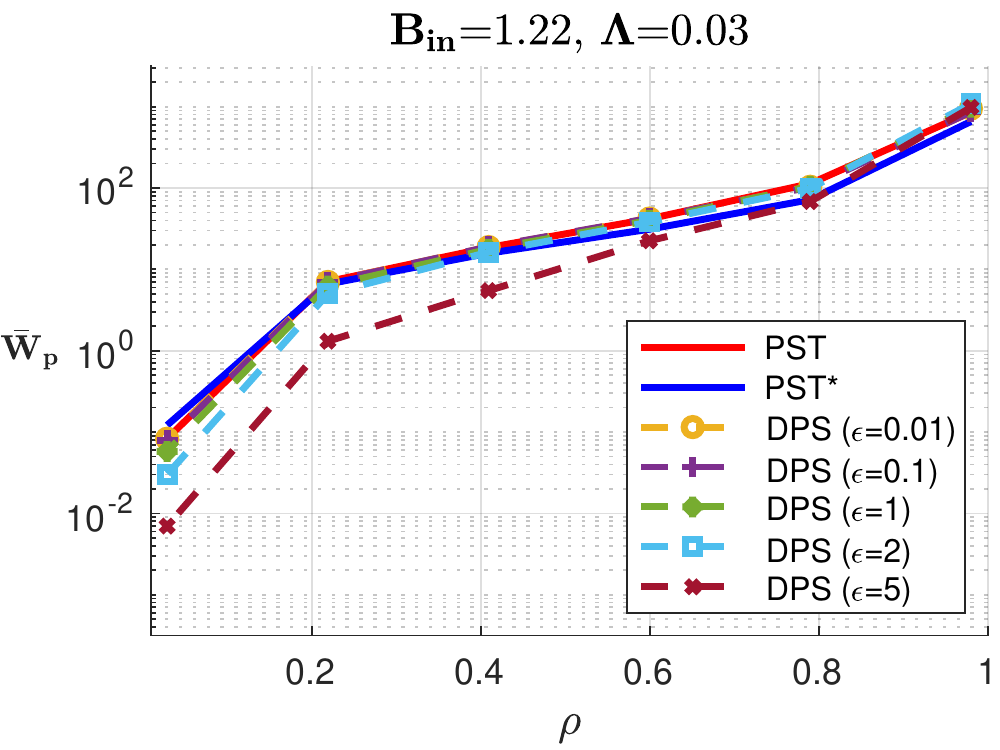}
			\caption{$s = 5$.}
			\label{fig:delay_dps_vs_vsci_s_5}
		\end{subfigure}
		\caption{$\bar{Q}_t/\bar{W}_p-\rho-\epsilon$ tradeoffs ($(\epsilon,\epsilon)$-DPS vs. PST/PST*) under Zipf PMFs with different scale parameters $s$.}
		\label{fig:dps_vs_vsci}
	\end{figure*}
	
	\begin{figure*}
		\centering
		\begin{subfigure}{\subfigwd}
			\centering
			\includegraphics[width=\linewidth,height=\subfight]{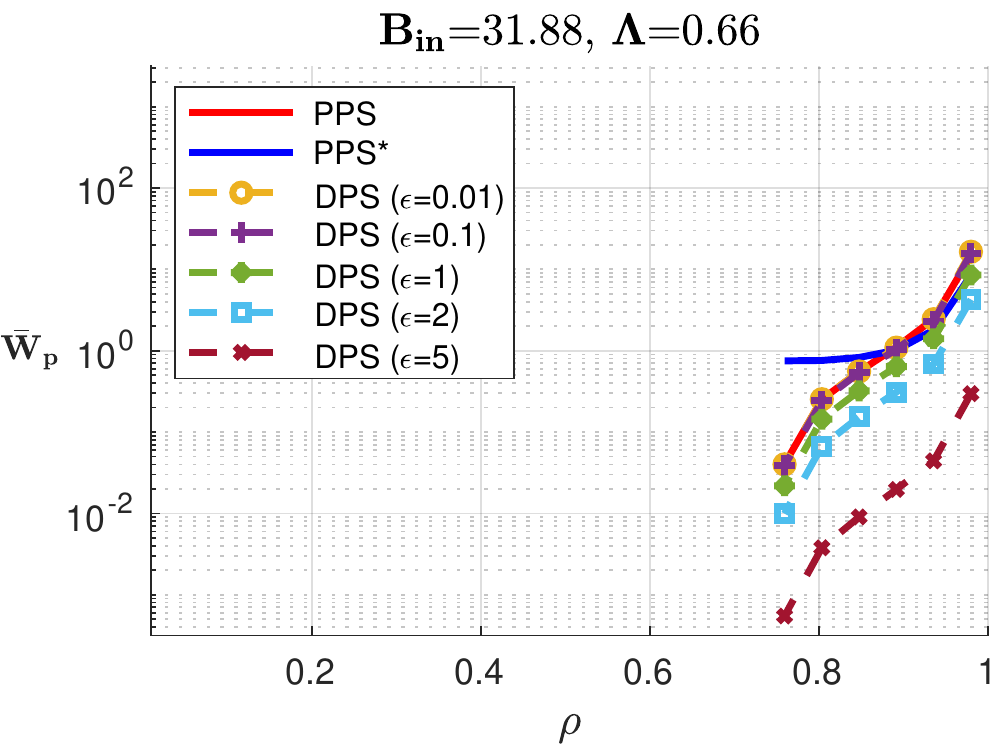}
			\caption{$s = 0.01$.}
			\label{fig:delay_dps_vs_vsvi_s_0_01}
		\end{subfigure}
		\begin{subfigure}{\subfigwd}
			\centering
			\includegraphics[width=\linewidth,height=\subfight]{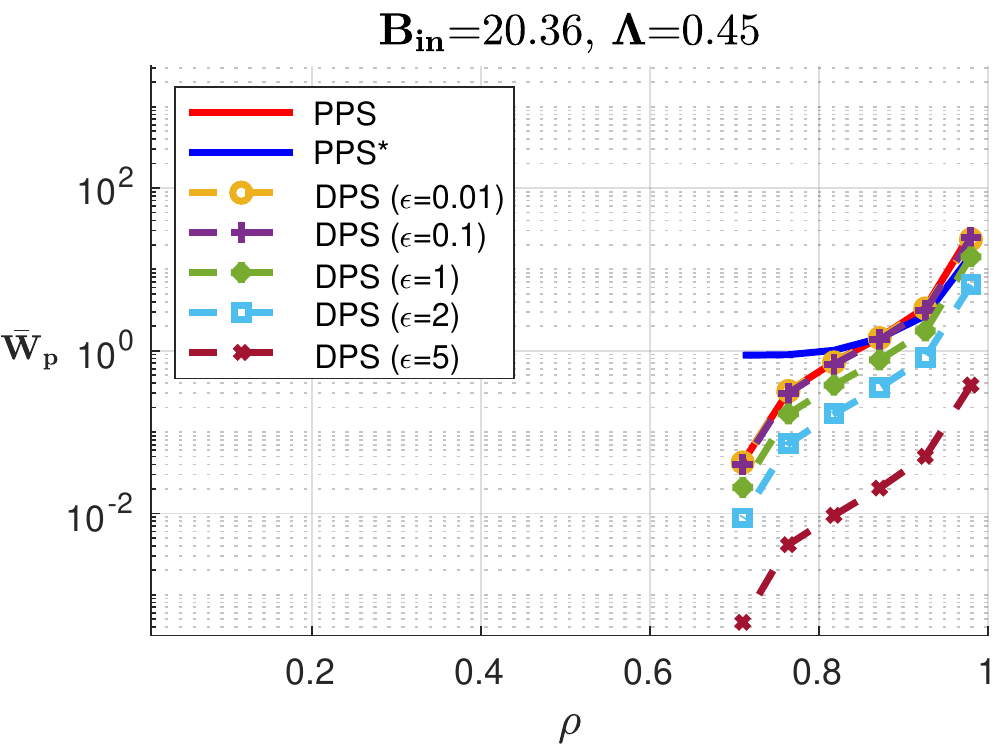}
			\caption{$s = 1$.}
			\label{fig:delay_dps_vs_vsvi_s_1}
		\end{subfigure}
		\begin{subfigure}{\subfigwd}
			\centering
			\includegraphics[width=\linewidth,height=\subfight]{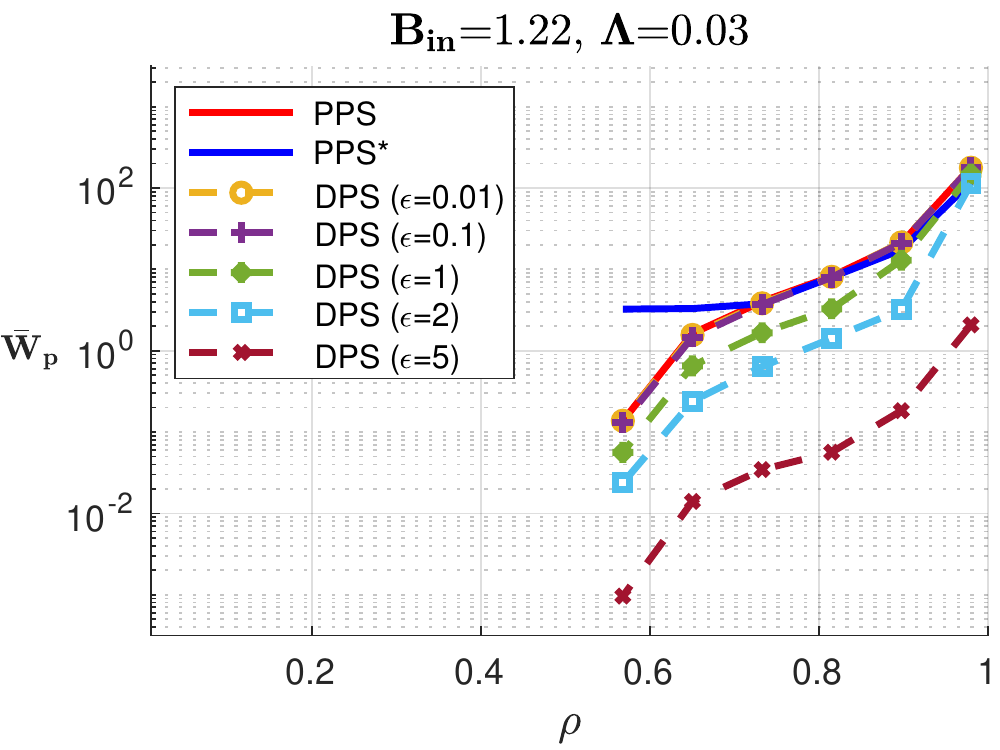}
			\caption{$s = 5$.}
			\label{fig:delay_dps_vs_vsvi_s_5}
		\end{subfigure}
		\caption{$\bar{W}_p-\rho-\epsilon$ tradeoffs ($(\epsilon,\infty)$-DPS vs. PPS/PPS*) under Zipf PMFs with different scale parameters $s$.}
		\label{fig:dps_vs_vsvi}
	\end{figure*}
	
	In Section~\ref{sec:EXP}, we will evaluate the effects of both PST and PST* (PPS and PPS*) mechanisms as baseline approaches to guaranteeing perfect event-level privacy (\emph{resp.} perfect event-level packet-size privacy). In the perfect-privacy regime, there naturally exists a tradeoff between the transmission efficiency level $\rho$ and the expected queue size $\mbb{E}[Q_t]$. Here, we show such tradeoffs for PST/PST* and PPS/PPS* mechanisms in Fig.~\ref{fig:tradeoff_VSCI_VSVI} for one of the experimental setups from Section~\ref{sec:EXP}. We make the following observations,
	\begin{itemize}
		\item More dummy traffic helps deplete the queue: decreasing $\rho$ leads to decreasing $\mbb{E}[Q_t]$. 
		\item The deterministic policies PST*/PPS* (black asterisk/blue cross lines) indeed result in smaller average queue sizes than the non-deterministic policies PST/PPS (red circle/green plus lines).
		\item PPS/PPS* achieve a less restrictive privacy guarantee than PST/PST* by introducing some dependency between the input and output. This yields significant savings on both the delay (smaller queue $\mbb{E}[Q_t]$) and byte rate overhead (higher efficiency $\rho$).
	\end{itemize}
	
	As the DPS mechanism interpolates between the perfect-privacy PST and PPS shapers and the non-private shaper (e.g., with an identity channel matrix), we are interested in how its privacy-overhead tradeoffs compare to those of the baseline approaches, and will show the comparisons in the sequel.
	
	\def\subfigwd{.31\linewidth}
	\def\subfight{.75\linewidth}
	\begin{figure*}
		\centering
		\begin{subfigure}{\subfigwd}
			\centering
			\includegraphics[width=\linewidth,height=\subfight]{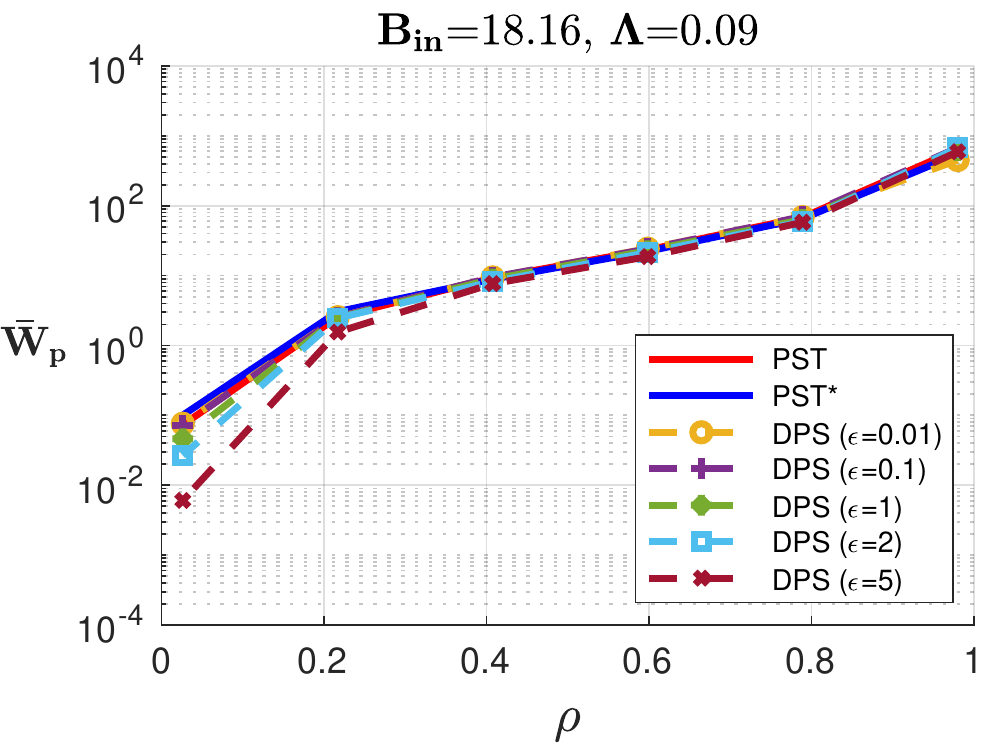}
			\caption{Sleep.}
			\label{fig:sleep}
		\end{subfigure}
		\begin{subfigure}{\subfigwd}
			\centering
			\includegraphics[width=\linewidth,height=\subfight]{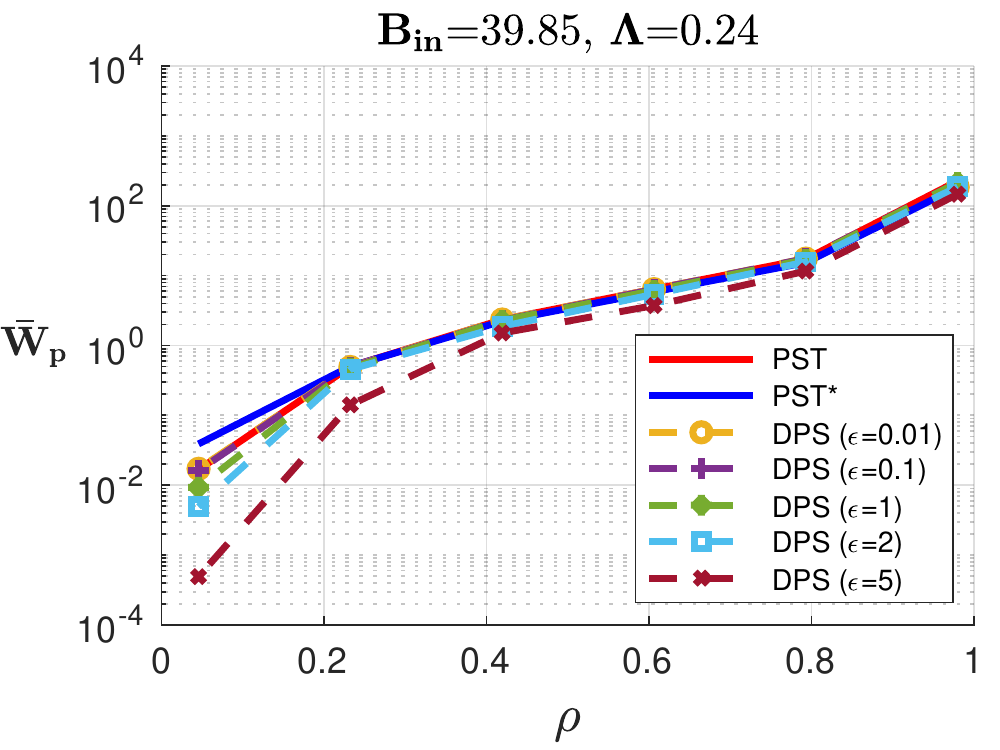}
			\caption{Sleep $+$ Camera.}
			\label{fig:cam}
		\end{subfigure}
		\begin{subfigure}{\subfigwd}
			\centering
			\includegraphics[width=\linewidth,height=\subfight]{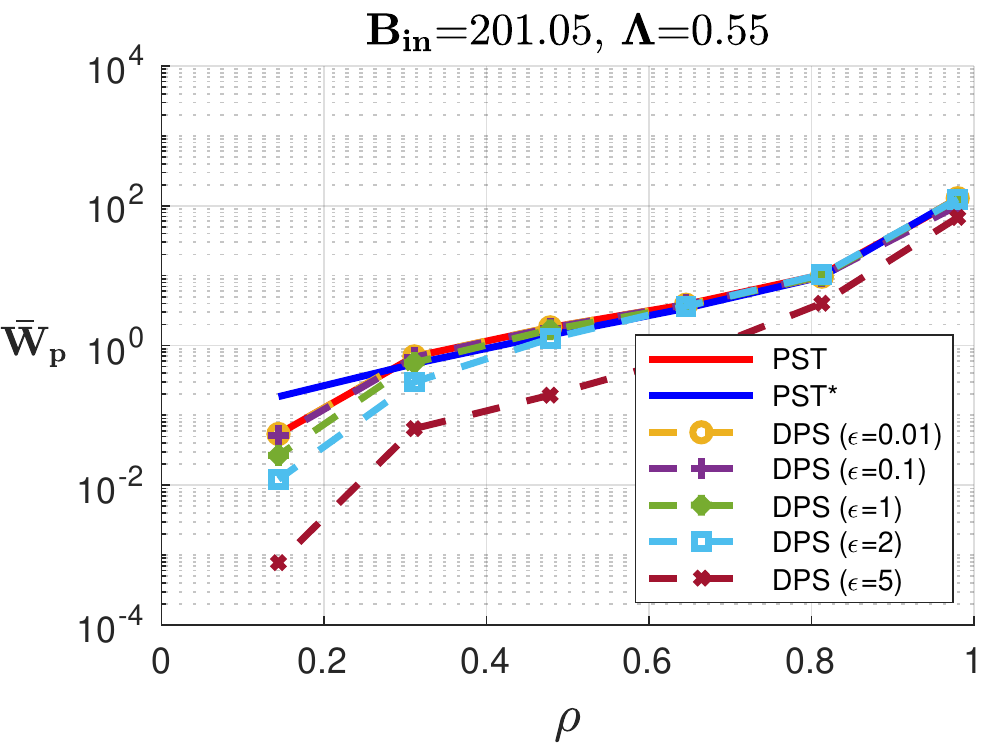}
			\caption{Sleep $+$ Camera $+$ Switch.}
			\label{fig:merge}
		\end{subfigure}
		\caption{Effect of increasing number of IoT devices on the $\bar{W}_p-\rho-\epsilon$ tradeoffs ($(\epsilon,\epsilon)$-DPS vs. PST/PST*).}
		\label{fig:merge_devices}
	\end{figure*}
	
	\subsection{Pad-Only Policy}
	For extremely low-latency networks and delay-sensitive traffic (e.g., smart healthcare devices monitoring the health conditions of users should have timely communication with the intended application servers), we prefer \emph{zero-delay} shaping mechanisms to protect privacy by adopting the \emph{pad-only} policy: $D_t\geq A_t, \forall t$. Here, we take a closer look at the effect of enforcing the pad-only policy on the PST, PPS and DPS mechanisms, and denote them as PST$^0$, PPS$^0$ and DPS$^0$, respectively, with the superscript marking $0$ delay.
	
	\subsubsection{PST$^0$ \& PPS$^0$} \label{sec:pad_only_pst_pps}
	It's easy to see that the PST$^0$ and PPS$^0$ mechanisms essentially output packets with the largest size $D_t=a_n$ for $\forall t\in\mc{T}$ and $\forall t\in\mc{I}$, respectively. That is, $\mu(d)=\upsilon(d)=\delta(a_n)$. As a result, their transmission efficiency levels in~\eqref{eqn:efficiency_pst} and~\eqref{eqn:efficiency_pps} change to,
	\begin{align}
		\rho^\textrm{PST$^0$}&=\mbs{\lambda}^\top\mbs{a}/a_n, \label{eqn:efficiency_pst_pad} \\
		\rho^\textrm{PPS$^0$}&=\mbs{\lambda}^\top\mbs{a}/a_n\cdot\frac{1}{\Lambda}. \label{eqn:efficiency_pps_pad}
	\end{align}
	
	We see that the PST$^0$ shaper is very transmission inefficient, and more so when the incoming traffic is a mouse flow (i.e., the arrival rate of event packets $\Lambda=\sum_{i\in N^+}\lambda_i$ is small which lowers the numerator and hence $\rho^\textrm{PST$^0$}$) than an elephant flow. This is intuitive since a mouse flow represents a scenario where private user activities only happen sporadically across time (contrary to an elephant flow) and presents more ``variability'' in the network traffic. An adversary can exploit the variability for more effective inference attacks. Shaping a mouse flow to guarantee perfect event-level privacy then becomes more costly in terms of dummy traffic. 
	
	The same is not necessarily true for the PPS$^0$ shaper since it preserves the timing information and the demand for dummy traffic is less affected by $\Lambda$ (a small $\Lambda$ reduces both the numerator and denominator in~\eqref{eqn:efficiency_pps_pad} where $\rho^\textrm{PPS$^0$}$ can increase or decrease). We will empirically validate this phenomenon in Section~\ref{sec:EXP} and show that it is generally true for privacy-preserving shapers with or without the pad-only policy. 
	
	\subsubsection{DPS$^0$}
	
	Since the input and output alphabets $\mc{A}$ and $\mc{D}$ are ordered in increasing sizes, enforcing the pad-only policy on the DPS mechanism means adding a structural constraint on the channel matrix $\mbs{C}$, that we only allow non-zero entries in the ``upper triangle'' (assuming $\mc{D}=\mc{A}$ w.l.o.g.),
	\begin{align}
		c_{ij}=c(d_j|a_i) = 0,\ \forall d_j < a_i. \label{eqn:padding}
	\end{align}
	Note that $\max\left(\epsilon_s, \frac{\epsilon_t}{2}\right)$-LDP bounds the pair-wise ratios in each column of $\mbs{C}$ by either $\epsilon_s$ or $\frac{\epsilon_t}{2}$ according to~\eqref{eqn:priv_constr_alt}. For finite $\epsilon_s,\epsilon_t<\infty$, if any entry in a column is 0, then that whole column has to be 0, otherwise $\epsilon_s, \epsilon_t$ become $\infty$. Along with the structural and right stochastic constraints, the channel matrix $\mbs{C}$ for a pad-only $(\epsilon_s,\epsilon_t)$-DP shaper can only have all ones in the last column, pushing $\epsilon_s,\epsilon_t$ to be 0. The $(0,0)$-DPS$^0$ mechanism then becomes identical to the PST$^0$ mechanism. By analogy, the $(0,\infty)$-DPS$^0$ mechanism is equivalent to the PPS$^0$ shaper. The DPS mechanism encapsulates both PST and PPS mechanisms with or without the pad-only policy.
	
	In the following section, we will evaluate the privacy-overhead tradeoffs of DPS, PST/PST*, PPS/PPS* mechanisms on different types of traffic under various settings of privacy parameters $(\epsilon_s,\epsilon_t)$ and transmission efficiency level $\rho$. We need not explore values of $\rho$ lower than $\rho^{\textrm{PST}^0}$ or $\rho^{\textrm{PPS}^0}$ which already yield 0 delay for the PST or PPS shaper.
	
	\section{Experimental Results}\label{sec:EXP}
	
	We experiment\footnote{The simulation code is posted at \href{https://github.com/sijie-xiong/PERMIT}{https://github.com/sijie-xiong/PERMIT}.} on synthetic data and packet traces from 3 smart home IoT devices (Sense Sleep monitor, Nest camera and WeMo switch)~\cite{apthorpe2017smart}. For synthetic data, we simulate packet traces with sizes drawn i.i.d. from Zipf distribution with PMF: $P_\mathrm{Zipf}(k;s,E)=(1/k^s)/\sum_{e=1}^E (1/e)^s$, which characterizes the frequency of rank-$k$ element out of a population of $E$ elements. We assume that packet size $a_i$ has rank $i+1, \forall i\in N$. We choose the exponent, or the scale parameter $s\in[0.01, 1, 5]$ for Zipf PMF, and set the possible packet sizes to be $\mc{A}=[0,32,64]$. In this work, we assume $\mc{D}=\mc{A}$ and leave for future work the design and optimization of $\mc{D}$. For IoT devices, we discretize the packet traces into 1s time slots and keep only the event packet sizes (e.g., 270B and 142B packets from the Nest camera triggered by motion detection and checking camera feed, respectively). We then extract packet size PMFs\footnote{That $\mc{A}_\textrm{Sleep}=[0, 93, 1117]$, $\lambda_\textrm{Sleep}=[0.91, 0.08, 0.01]$; $\mc{A}_\textrm{Camera}=[0, 142, 270]$, $\lambda_\textrm{Camera}=[0.85, 0.14 , 0.01]$; and $\mc{A}_\textrm{Switch}=[0, 40, 1500]$, $\lambda_\textrm{Switch}=[0.69, 0.21, 0.1]$.} from the preprocessed traces for optimizing the shapers.
	
	By setting different values of the scale parameter $s$ in Zipf PMF, we are essentially synthesizing packet streams ranging from mouse to elephant flows. For smaller $s$, the Zipf PMF has a heavier tail, putting more probability mass on larger packets, thus creating heavier input traffic that exemplifies an elephant flow (higher $\Lambda=1-\lambda_0$). An elephant flow with small $s$ also identifies increased traffic from a larger number of IoT devices, and vice versa.
	
	\subsection{Empirical Privacy-Overhead Tradeoffs}
	In Section~\ref{sec:POT}, we formally characterized the privacy-overhead tradeoff in Theorem~\ref{theorem:POT} between the minimum achievable $\mbb{E}[Q_t]$, transmission efficiency $\rho$ and privacy levels $(\epsilon_s,\epsilon_t)$. To validate their relationship, as well as compare the empirical privacy-overhead tradeoffs of different shapers, we will solve~\eqref{opt:DPS},~\eqref{opt:PST} and~\eqref{opt:PPS} under varying $\rho$, $(\epsilon_s,\epsilon_t)$ and packet size PMFs to find the optimal distributions $\mbs{C}^*$, $\mbs{\mu}^*$ and $\mbs{\upsilon}^*$. We then use discrete-event simulation~\cite{matloff2008introduction} to calculate the empirical overhead measures $\bar{W}_p$~\eqref{eqn:avg_delay_per_pkt} and $\bar{Q}_t$~\eqref{eqn:avg_qsz_across_time} after running the optimized shapers DPS ($\mbs{C}^*$), PST ($\mbs{\mu}^*$) and PPS ($\mbs{\upsilon}^*$) on packet traces (synthetically generated or from IoT devices) for sufficiently large $T=100000$ time slots. We note that the empirical measure $\bar{Q}_t$ is always consistent with the expected value $\mbb{E}[Q_t]$ estimated by WHF\footnote{The red circle/black asterisk lines in Fig.~\ref{fig:tradeoff_VSCI_VSVI} measuring $\mbb{E}[Q_t]$ match the red/blue solid lines in Fig.~\ref{fig:qsz_dps_vs_vsci_s_5} measuring $\bar{Q}_t$ for PST/PST* shapers.} method, we omit the matching results due to space limit. 
	
	In what follows, we let $\epsilon\in[0.01,0.1,1,2,5]$ and compare:
	\begin{itemize}
		\item $(\epsilon,\epsilon)$-DPS vs. PST/PST* shapers for $\rho\in(\rho^{\textrm{PST}^0},1)$\footnote{Here, we consider $\epsilon_s=\epsilon_t=\epsilon$ for the DPS mechanism due to space limit. For the effect of setting different $\epsilon_s$ and $\epsilon_t$ values on the $\bar{W}_p-\rho-(\epsilon_s,\epsilon_t)$ tradeoffs, please refer to the supplementary material.},
		\item $(\epsilon,\infty)$-DPS vs. PPS/PPS* shapers for $\rho\in(\rho^{\textrm{PPS}^0},1)$,
	\end{itemize}
	in terms of their empirical privacy-overhead (i.e., $\bar{Q}_t$/$\bar{W}_p$-$\rho$-$\epsilon$) tradeoffs, with results shown in Fig.~\ref{fig:dps_vs_vsci} and Fig.~\ref{fig:dps_vs_vsvi}, respectively. We stress that the lower (smaller $\bar{Q}_t$/$\bar{W}_p$) and closer to the right (higher $\rho$) the tradeoff curves are, the less shaping overhead the mechanisms require for privacy protection. We then make the following observations,
	\begin{enumerate}
		\item The general trend of $\bar{Q}_t$-$\rho$-$\epsilon$ tradeoffs (Fig.~\ref{fig:qsz_dps_vs_vsci_s_0_01}-\ref{fig:qsz_dps_vs_vsci_s_5}) retains in the $\bar{W}_p$-$\rho$-$\epsilon$ tradeoffs (Fig.~\ref{fig:delay_dps_vs_vsci_s_0_01}-\ref{fig:delay_dps_vs_vsci_s_5}). Minimizing $\mbb{E}[Q_t]$ indeed serves as a good proxy for minimizing $\mbb{E}[W_p]$.
		\item $\bar{Q}_t$/$\bar{W}_p$ increases when $\epsilon$ decreases with $\rho$ fixed and when $\rho$ increases with $\epsilon$ fixed. This verifies the privacy-overhead tradeoff characterized in Theorem~\ref{theorem:POT}. Guaranteeing $(\epsilon,\infty)$-DP (Fig.~\ref{fig:delay_dps_vs_vsvi_s_0_01}/\ref{fig:delay_dps_vs_vsvi_s_1}/\ref{fig:delay_dps_vs_vsvi_s_5}) instead of $(\epsilon,\epsilon)$-DP (\emph{resp.} Fig.~\ref{fig:delay_dps_vs_vsci_s_0_01}/\ref{fig:delay_dps_vs_vsci_s_1}/\ref{fig:delay_dps_vs_vsci_s_5}) vastly reduces the shaping overhead.
		\item Shapers become less overhead-efficient to guarantee the same level of privacy for a mouse flow than an elephant flow. As the input traffic changes from an elephant flow to a mouse flow (i.e., the scale parameter $s$ of the Zipf PMF increases from the left to the right plots), the shaping overhead increases (i.e., the tradeoff curves get higher and closer to the left).
		\item The advantage of $(\epsilon,\epsilon)$-DPS over PST/PST* in terms of preserving shaping overhead is more evident when we're dealing with an elephant flow than a mouse flow. That the curves in Fig.~\ref{fig:qsz_dps_vs_vsci_s_0_01}/\ref{fig:delay_dps_vs_vsci_s_0_01} are further apart\footnote{Since the y-axis is in log scale, small gaps between curves actually indicate differences in orders of magnitude.} than those in Fig.~\ref{fig:qsz_dps_vs_vsci_s_5}/\ref{fig:delay_dps_vs_vsci_s_5}. Intuitively, it's easier to hide among heavier traffic than lighter traffic. 
		\item The advantage of $(\epsilon,\infty)$-DPS over PPS/PPS*, however, behaves in the opposite fashion. The relative distances between the tradeoff curves now increases from Fig.~\ref{fig:delay_dps_vs_vsvi_s_0_01} to~\ref{fig:delay_dps_vs_vsvi_s_5}). The arrival rate $\Lambda$ of the input traffic becomes less impactful on the privacy-overhead tradeoffs of the shapers that only protect packet sizes information.
	\end{enumerate}
	The last 3 observations extend the same arguments for shapers with pad-only policy in Section~\ref{sec:pad_only_pst_pps}. 
	
	\subsection{Increasing Number of IoT Devices}
	
	As observed, heavier input traffic to the shapers yields better privacy-overhead tradeoffs. We want to understand whether increased traffic from a larger number of IoT devices will have the same effect. We therefore optimize the shapers on packet size PMF from a single device (e.g., Sense Sleep monitor), and the aggregated PMFs\footnote{For example, given $\mc{A}_\textrm{Sleep}=[0, 93, 1117]$, $\lambda_\textrm{Sleep}=[0.91, 0.08, 0.01]$ and $\mc{A}_\textrm{Camera}$ $=[0, 142, 270]$, $\lambda_\textrm{Camera}=[0.85, 0.14 , 0.01]$, their merged PMF is $\mc{A}_\textrm{Sleep + Camera}=$ $[0, 93, 142, 270, 1117]$, $\lambda_\textrm{Sleep + Camera}=[0.76, 0.08, 0.14,$ $0.01, 0.01]$ and the merged arrival rate $1-0.76$ equals the sum of Sense Sleep monitor's and Nest camera's arrival rates $(1-0.91)+(1-0.85)$.} from more devices, and plot their privacy-overhead tradeoffs in Fig.~\ref{fig:merge_devices}. We can see that the tradeoff curves get lower and closer to the right from Fig.~\ref{fig:sleep} to Fig.~\ref{fig:merge}. The shapers optimized for traffic aggregated from more IoT devices require less shaping overhead than those optimized for individual traffic flows.
	
	\begin{figure}[t]
		\centering
		\begin{subfigure}{\linewidth}
			\centering
			\includegraphics[width=0.6\linewidth,height=.45\linewidth]{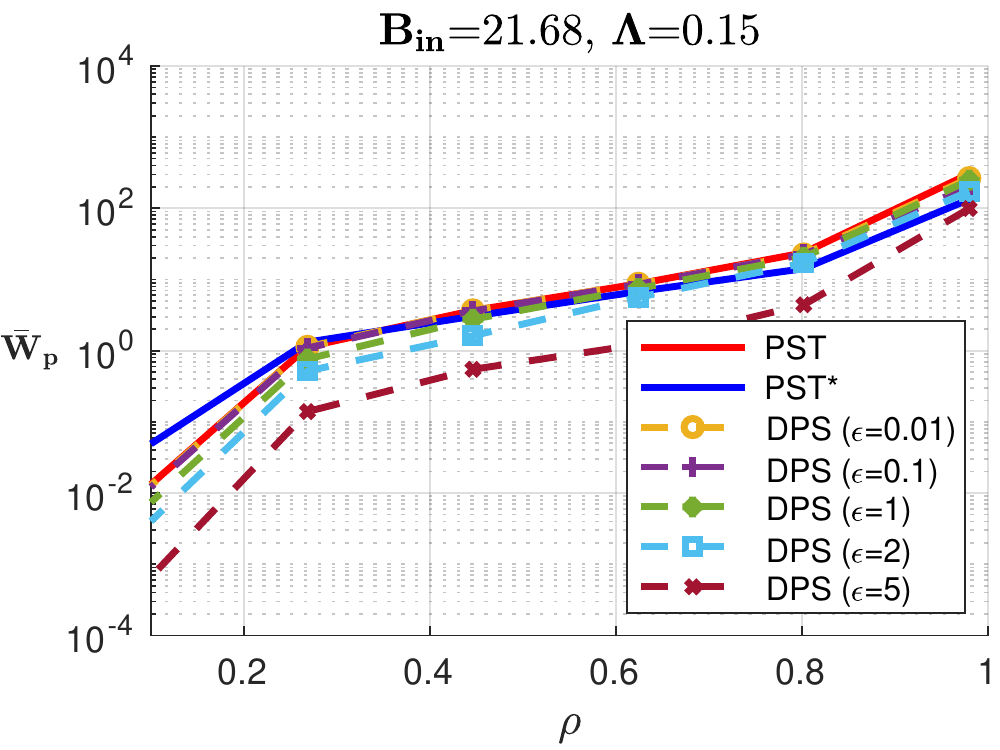}
			\caption{Performed on synthetically and i.i.d. generated packet stream.}
			\label{fig:iid}
		\end{subfigure}\\
		\begin{subfigure}{\linewidth}
			\centering
			\includegraphics[width=0.6\linewidth,height=.45\linewidth]{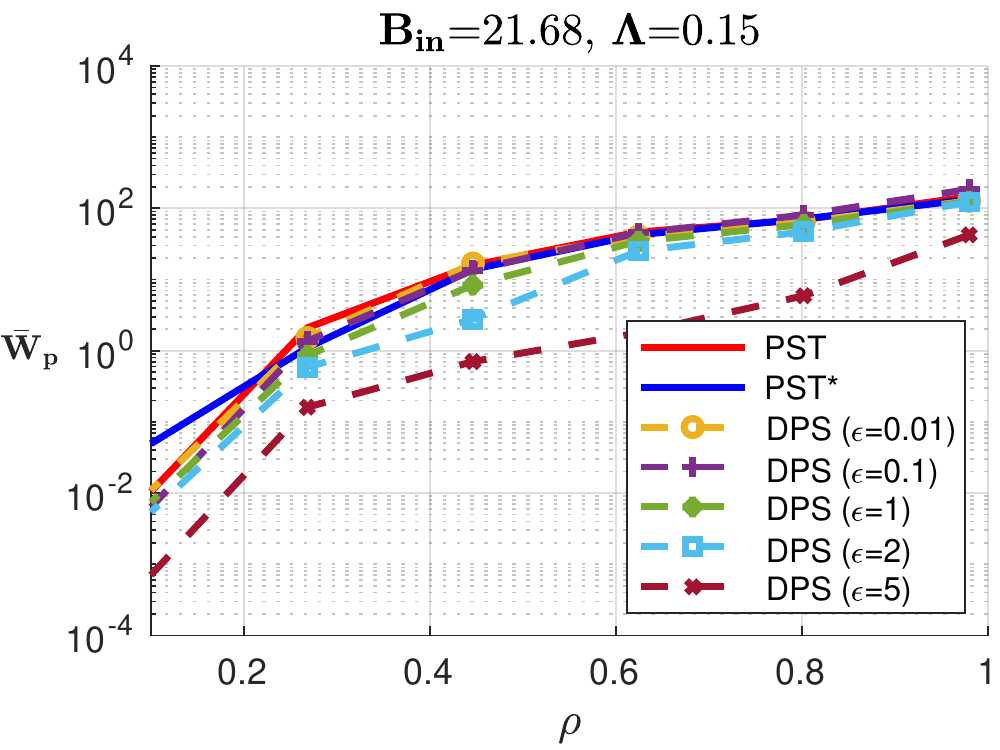}
			\caption{Performed on original bursty traffic.}
			\label{fig:bursty}
		\end{subfigure}
		\caption{$\bar{W}_p-\rho-\epsilon$ tradeoffs ($(\epsilon,\epsilon)$ DPS vs. PST/PST*) on Nest camera's packet streams (i.i.d. vs. bursty).}
		\label{fig:iid_vs_bursty}
	\end{figure}
	
	\begin{figure*}
		\centering
		\begin{subfigure}{0.245\linewidth}
			\centering
			\includegraphics[width=\linewidth,height=.7\linewidth]{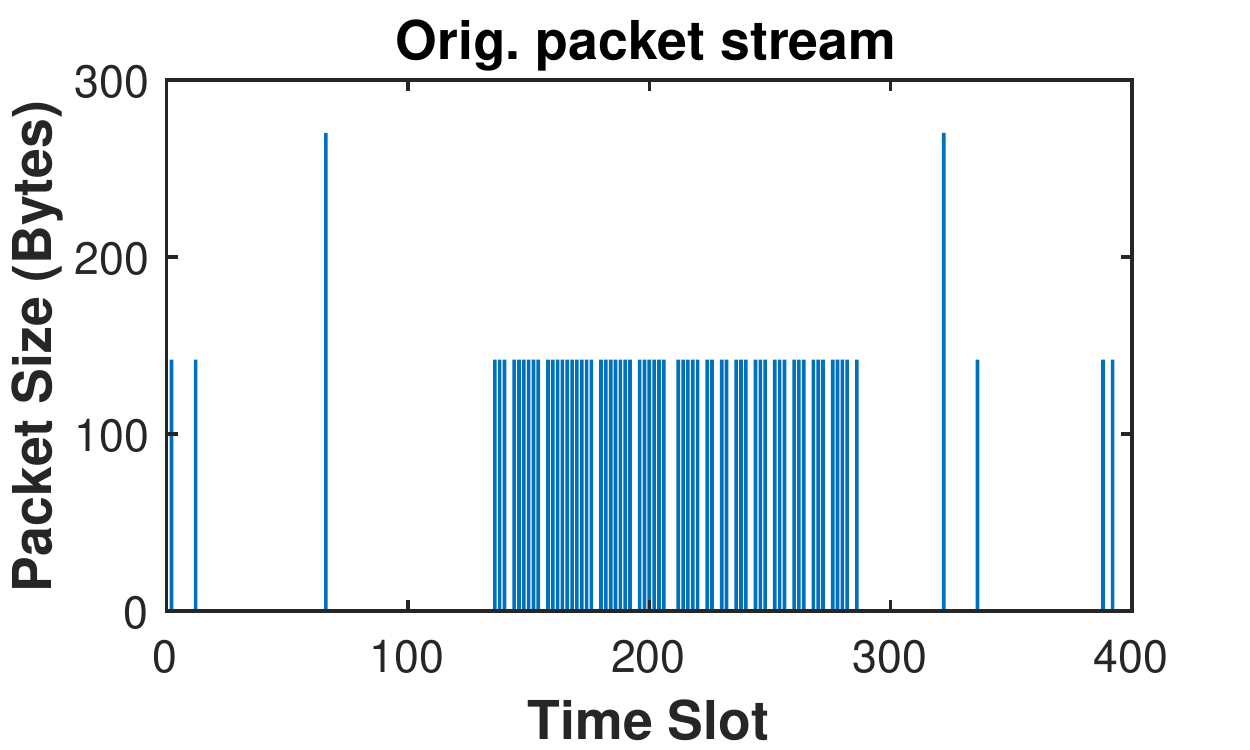}
			\caption{Original packet stream.}
			\label{fig:bursty_orig}
		\end{subfigure}
		\begin{subfigure}{0.245\linewidth}
			\centering
			\includegraphics[width=\linewidth,height=.7\linewidth]{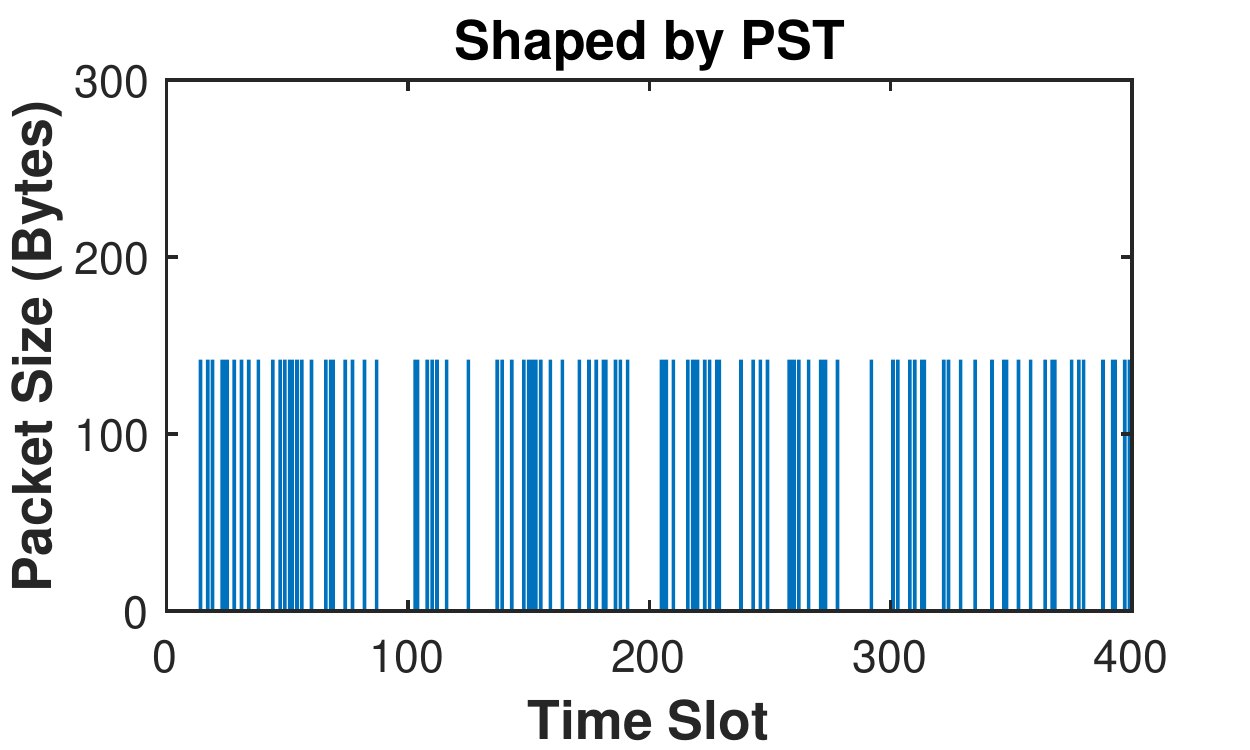}
			\caption{Shaped by PST.}
			\label{fig:bursty_vsci}
		\end{subfigure}
		\begin{subfigure}{0.245\linewidth}
			\centering
			\includegraphics[width=\linewidth,height=.7\linewidth]{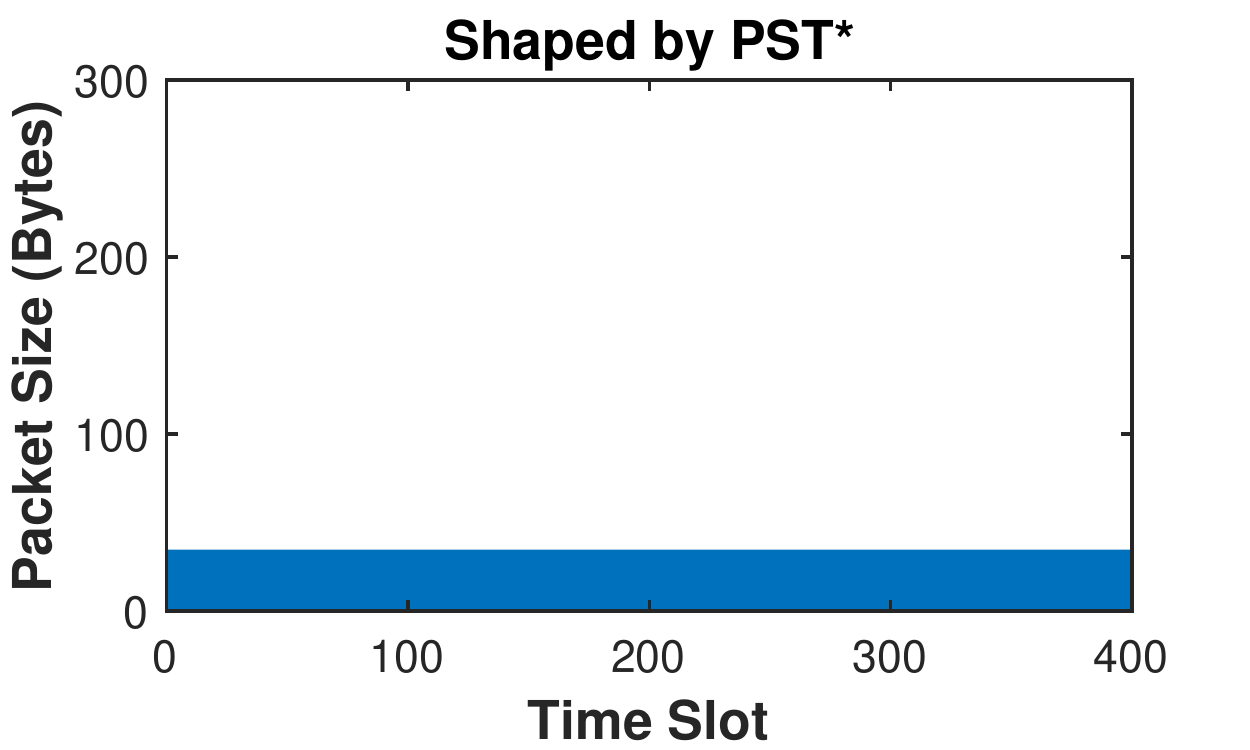}
			\caption{Shaped by PST*.}
			\label{fig:bursty_csci}
		\end{subfigure}
		\begin{subfigure}{0.245\linewidth}
			\centering
			\includegraphics[width=\linewidth,height=.7\linewidth]{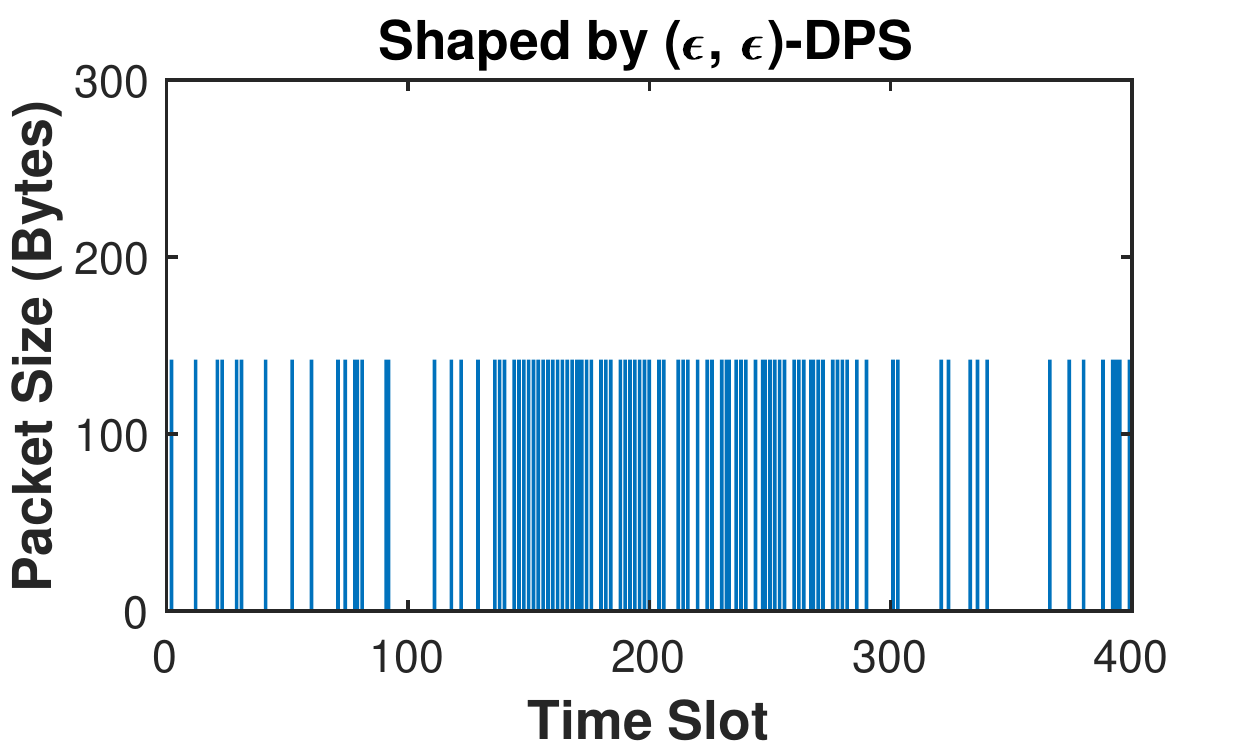}
			\caption{Shaped by $(\epsilon,\epsilon)$-DPS.}
			\label{fig:bursty_dps}
		\end{subfigure}
		\caption{Comparison between the original event packet stream from Nest camera and shaped traffic by PST, PST*, and $(\epsilon,\epsilon)$-DPS mechanisms. In particular, we set $\rho=0.62$ for all the shapers and $\epsilon=5$ for the DPS mechanism.}
		\label{fig:pkt_trace_before_n_after}
	\end{figure*}
	
	\subsection{Performance on Bursty Traffic} \label{sec:exp_corr}
	When we shape real IoT traffic that is often unpredictable or bursty using shapers optimized for a different source distribution, a critical aspect for evaluating the effect of shaping is to see how the privacy-overhead tradeoff changes. We discussed in Section~\ref{sec:DP_ADV} that shapers designed with DP maintain the same privacy guarantees despite the change in source distributions (e.g., from i.i.d. to bursty). Yet it remains to see how such change affects the amount of shaping overhead.
	
	To this end, we solve for the optimal DPS and PST/PST* mechanisms based on Nest camera's packet size PMF\footnote{We omit results from the other 2 IoT devices and PPS/PPS* shapers due to their similarities in the change of privacy-overhead tradeoffs.}. We then run 2 packet streams through the optimized shapers: one is synthetically generated by i.i.d. draws from the PMF and the other one is the original bursty traffic. We plot the corresponding privacy-overhead tradeoffs in Fig.~\ref{fig:iid} and~\ref{fig:bursty}, respectively. We see that traffic burstiness doesn't hurt the delay overhead by much when the shapers have enough dummy traffic at disposal ($\rho<0.25$), but creates longer backlogs in the queue for all the shapers with limited dummy traffic ($\rho>0.25$). In the latter region, however, the DPS mechanism introduces less additional delay overhead than PST/PST* shapers.

	\subsection{Event Packet Streams Before/After Traffic Shaping}
	\label{sec:before_after}
	
	In Fig.~\ref{fig:pkt_trace_before_n_after}, we visualize the original event packet stream of Nest camera against its shaped traffic in the same window of 200 time slots by running it through the PST, PST* and $(\epsilon,\epsilon)$-DP shapers, respectively. We see from Fig.~\ref{fig:bursty_orig} that in time slots $\sim$70 and $\sim$320, there are two 270B packets corresponding to motion detection, and from time slot $\sim$130 to 280, there's a long burst of 142B packets indicating a time period when the user is checking the camera feed. From the output of PST/PST* shapers in Fig.~\ref{fig:bursty_vsci} and~\ref{fig:bursty_csci}, we couldn't detect the presence of such event-indicating traffic. In Fig.~\ref{fig:bursty_dps}, the output of the $(\epsilon,\epsilon)$-DP shaper with $\epsilon=5$ obscures the events of motion detection, though reveals the period of checking camera feed to some extent. This highlights the limitation of the event-level DP model on infinite packet streams.
	
	\section{Limitations and Future Directions} \label{sec:L&F}
	One way to overcome the limitation of the event-level DP model in hiding long bursts of event packets is to extend the definition of event packet stream. That is, we can expand from single-packet events to burst-of-packets events by buffering the bursts and merge the enclosed packets, and the length of an indivisible time slot can be specified as the maximum duration of event bursts. We can then optimize and run a DP shaper on such redefined event burst stream. This is like running the clock for the gateway output at a slower cycle than the network behind it -- the shaper generates packets at a slower rate (but with bigger payloads), imposing a fixed amount of delay (the cycle length) at baseline. A downside of this approach is that the cycle length as well as the onset times of event bursts are subject to the user's real-time interactions with IoT devices and they are oftentimes impossible to know in the design phase.
	
	Another interesting extension of the event-level privacy is to look at the $w$-event privacy model~\cite{kellaris2014differentially} over infinite packet streams, which protects any \emph{event
		sequence} occurring in a window of $w$ consecutive time slots. It can protect any temporally constrained user activities from being disclosed, just as the period of checking camera feed. This approach is rather different from the extension to event bursts since the $w$-event privacy model protects any event burst (no longer than length $w$) irrespective of its onset time.
	
	In fact, by sequential composition~\cite{mcsherry2009privacy}, an $(\epsilon, \epsilon)$-DP shaper can be shown to guarantee $w$-event $w\cdot\epsilon$-DP, yet the privacy leakage grows linearly with the length of the time window $w$. One way to address this is to design a shaping mechanism \emph{with memory}: $P(D_t|A^t,D^{t-1})=P(D_t|A^{[t-w+1:t]},D^{[t-w+1:t-1]})$, so that sublinear privacy leakage is achievable by carefully chosen departures accounted for the dependency across time windows of length $w$. Yet this is challenging due to a much higher dimensional optimization space.
	
	Hiding long bursts of events falls under the broader subject of dealing with correlated traffic. The optimization of our proposed DPS mechanism relies on the assumption that $A^T$ is i.i.d. across time to be efficiently solved by convex programming and the WHF method. However, correlation in the input traffic can potentially be modeled and utilized to further improve the overhead efficiency of shaper design. For example, if both the system designer and the adversary know that a user sleeps strictly between 8pm-11pm, thus triggering an event packet stream with 1117B packet (generated by Sense Sleep monitor) only observable between 8pm-11pm, then a time-constrained shaper can be designed to only shape the traffic between 8pm-11pm for reduced overhead.
	
	The time window exemplifies the \emph{constraint specification} in the framework of Blowfish privacy~\cite{he2014blowfish} to restrict the set of realizable $A^T$ due to correlation. Comparably, $A^T$ can also be drawn from a \emph{distribution class} (e.g., Markov chains) following Pufferfish privacy~\cite{song2017pufferfish} for which more overhead-efficient shapers can possibly be designed. However, these privacy models will be hard to use in practice. On one hand, it would be very device/user dependent and then the privacy and overhead guarantees would be contingent on the device/user behaving \emph{typically}. On the other hand, estimating the correlation from traffic data and optimizing shapers may become computationally expensive and even intractable.
	
	Another assumption that our proposed shaping mechanism relies on is an intermediate platform trusted and shared by many households to reverse the packet ``surgery''. To relax this assumption, we can design and deploy the DPS mechanism on a device level: it can be implemented to shape the outgoing network traffic of individual IoT devices. We can think of each device with its own dedicated FCFS queue and delay-optimal DP shaper. Designing a traffic shaping system in this way motivates the study of optimal allocation of privacy budget and network resources, as well as optimal scheduling of device outputs in a local area IoT network.
	
	\section{Conclusion} \label{sec:CON}
	In this work, we motivate the need for designing network traffic shaping in IoT networks under the framework of DP. We establish a rigorous event-level DP model on discrete event packet streams and propose an event-level $(\epsilon_s,\epsilon_t)$-DP shaping mechanism which utilizes a discrete memoryless $\max\left(\epsilon_s, \frac{\epsilon_t}{2}\right)$-LDP channel $c$ to protect both packet sizes and timing information from traffic analysis attacks. Under special settings of $(\epsilon_s,\epsilon_t)$ and deterministic/pad-only policies, the DPS mechanism becomes equivalent to previous shaping schemes proposed in other contexts. All shapers work by generating the output packet stream in ways that are either independent or dependent from the input traffic. The dependency introduced by the channel $c$ offers the DPS mechanism more degrees of freedom in trading off privacy for shaping overhead. 
	
	We empirically evaluate all shapers on synthetic data and packet traces from actual IoT devices. Under various types of input traffic, we discover interesting fundamental privacy-overhead tradeoffs: increased traffic from a larger number of IoT devices makes user privacy protection easier. The DPS mechanism not only enhances the privacy-overhead tradeoffs of the PST/PST* and PPS/PPS* mechanisms, but also handles bursty traffic better. This novel prototype for building a privacy-preserving and overhead-efficient traffic shaping system enables users to adapt to their privacy demands and network conditions. It serves as a foundation for understanding and defending against more sophisticated traffic analysis attacks with strong, formal and tunable privacy guarantee.
	
	
	%
	
	\appendix
	\subsection{Proof of Proposition~\ref{prop:DPS}} \label{appdix:DPS}
	\begin{proof}
		The set of constraints in~\eqref{eqn:priv_constr_alt} is equivalent to 
		\begin{align}
			\left\{
			\begin{array}{ll}
				\max\limits_{\substack{\forall i,k\in N^+\\ \forall j\in M}}\left\{\frac{c_{ij}}{c_{kj}}\right\} \leq e^{\epsilon_s}, \\
				\max\limits_{\substack{\forall i,k\in N^+\\ \forall j\in M}}\left\{\frac{c_{ij}}{c_{0j}}, \frac{c_{0j}}{c_{kj}}\right\} \leq e^\frac{\epsilon_t}{2}. 
			\end{array}\right.
			\Leftrightarrow
			\max\limits_{\substack{\forall i,k\in N\\ \forall j\in M}}\left\{\frac{c_{ij}}{c_{kj}}\right\}\leq e^{\max\left(\epsilon_s,\frac{\epsilon_t}{2}\right)}. \nonumber
		\end{align}
		According to Definition~\ref{defn:LDP}, the channel $c$ satisfying~\eqref{eqn:priv_constr_alt} is $\max\left(\epsilon_s, \frac{\epsilon_t}{2}\right)$-LDP.
		
		We then show the DP guarantee of the shaping mechanism $\mc{M}^\textrm{DPS}_t:\mc{A}\xrightarrow{c}\mc{D}$. Since the privacy model in Definition~\ref{defn:dp_ind} is defined on event-level adjacent packet stream prefixes, we look at the different adjacency pairs in Fig.~\ref{fig:adjacency}. For the event-level packet-size adjacency in Fig.~\ref{fig:adjacency_type}, $A^T$ and $\tilde{A}^T$ differ in a single time slot $t$ where $A_t\neq\tilde{A}_t>0$ and $A_s=\tilde{A}_s,\forall s\neq t\in [T]$, so we have $\epsilon_\textrm{DP}(\mc{M}^\textrm{DPS})$ as,
		\begin{align}
			& \left|\sum_{s=1, s\neq t}^T\log\frac{P[\mc{M}^\textrm{DPS}_s(A_s)=d_s]}{P[\mc{M}^\textrm{DPS}_s(\tilde{A}_s)=d_s]}+\log\frac{P[\mc{M}^\textrm{DPS}_t(A_t)=d_t]}{P[\mc{M}^\textrm{DPS}_t(\tilde{A}_t)=d_t]}\right|
			\nonumber\\=
			& \left|\log\frac{P[\mc{M}^\textrm{DPS}_t(A_t)=d_t]}{P[\mc{M}^\textrm{DPS}_t(\tilde{A}_t)=d_t]} \right| = \left|\log\frac{c(d_t|A_t>0)}{c(d_t|\tilde{A}_t>0)}\right|
			\nonumber\\\leq
			& \left|\log\max_{\forall i,k\in N^+, \forall j\in M}\frac{c(d_j|a_i)}{c(d_j|a_k)}\right| \leq \left|\log e^{\epsilon_s}\right| = \epsilon_s. \label{eqn:eps_s}
		\end{align}
		Likewise, for event-level packet-timing adjacent $A^T$ and $\tilde{A}^T$ in Fig.~\ref{fig:adjacency_timing} differing in 2 different time slots $s$ and $t$, we have $A_t=\tilde{A}_s>0$ and $\tilde{A}_t=A_s=0$. Then $\epsilon_\textrm{DP}(\mc{M}^\textrm{DPS})$ becomes,
		\begin{align}
			& \left|\log\frac{c(d_t|A_t>0)}{c(d_t|\tilde{A}_t=0)}+\log\frac{c(d_s|A_s=0)}{c(d_s|\tilde{A}_s>0)}\right| \nonumber \\
			\leq & \left|\log\max_{\forall i\in N^+, \forall j\in M}\frac{c(d_j|a_i)}{c(d_j|a_0)}\right| + \left|\log\max_{\forall k\in N^+, \forall j\in M}\frac{c(d_j|a_0)}{c(d_j|a_k)}\right| \nonumber \\
			\leq & 2\left|\log e^{\epsilon_t/2}\right| = \epsilon_t. \label{eqn:eps_t}
		\end{align}
		We have the first lines in~\eqref{eqn:eps_s} and~\eqref{eqn:eps_t} from the memoryless property~\eqref{eqn:dp_memoryless} and the last inequalities from the set of constraints~\eqref{eqn:priv_constr_alt}. By Definition \ref{defn:dp_ind} and \ref{defn:dp_all},  the shaping mechanism $\mc{M}^\textrm{DPS}_t:\mc{A}\xrightarrow{c}\mc{D}$ with a $\max\left(\epsilon_s, \frac{\epsilon_t}{2}\right)$-LDP memoryless channel $c$ satisfying~\eqref{eqn:priv_constr_alt} guarantees event-level $(\epsilon_s, \epsilon_t)$-DP. 
	\end{proof}
	
	\subsection{Proof of Proposition~\ref{prop:PST}} \label{appdix:PST}
	\begin{proof}
		For event-level packet-size or packet-timing adjacent prefixes $A^T$ and $\tilde{A}^T$, and any realization of the mechanism output $\mbf{d}^T\in\mc{D}^T$, we have $\epsilon_\textrm{DP}(\mc{M}^\textrm{PST})=\left|\sum_{t=1}^T\log\frac{\mu(d_t)}{\mu(d_t)}\right|=0$. By Definition \ref{defn:dp_ind}, PST guarantees perfect event-level privacy, or $(0,0)$-DP.
	\end{proof}
	
	\subsection{Proof of Proposition~\ref{prop:PPS}} \label{appdix:PPS}
	\begin{proof}
		For two event-level packet-size adjacent prefixes $A^T$ and $\tilde{A}^T$ according to Definition~\ref{defn:adjacency}, we have $\epsilon_\textrm{DP}(\mc{M}^\textrm{PST})=\abs{\sum_{t\in\mc{I}}\log \frac{\upsilon(d_t)}{\upsilon(d_t)}+\sum_{t\in \mc{T}\setminus\mc{I}}\log \frac{1}{1}}=0$. Hence the PPS mechanism guarantees perfect event-level packet-size privacy, or $(0,\infty)$-DP by Definition \ref{defn:dp_ind}.
	\end{proof}
	
	\subsection{Proof of Proposition~\ref{prop:CVX}} \label{appdix:CVX}
	\begin{proof}
		All the constraints in~\eqref{opt:DPS} are affine. To reason about the convexity of the objective $\mbb{E}[Q_t]$ as a function of $\mbs{C}$, we see that $X_t=A_t-D_t, t=1,2,\ldots$ are a sequence of i.i.d. random variables parameterized by the same stochastic mapping $\mbs{C}$. As a result, $\{X_t(\mbs{C})\}$ satisfies \emph{strong stochastic convexity} (SSCX) in $\mbs{C}$~\cite{shanthikumar1991strong}. The recursive Lindley's equation~\eqref{eqn:lindley_qsz2} involves only the $\max(x,0)$ function and the $+$ operator, both are increasing and convex functions. By the preservation of convexity~\cite{boyd2004convex}, $Q_t$ is an increasing and convex function of $(X_1, X_2,\ldots, X_t)$ parameterized by $\mbs{C}$, hence $\{Q_t(\mbs{C})\}$ satisfies SSCX in $\mbs{C}$ as well. SSCX means that for any increasing and convex function $f: \mbb{R}\rightarrow\mbb{R}$, including the identity function $f(x)=x$, $\mbb{E}[f(Q_t(\mbs{C}))]$ is convex in $\mbs{C}$. Therefore, $\mbb{E}[Q_t]$ is a convex function of $\mbs{C}$.
		
		The same argument applies to~\eqref{opt:PST} and~\eqref{opt:PPS}. As PST ($\mbs{\mu}$) and PPS ($\mbs{\upsilon}$) shapers are special cases of the DPS mechanism ($\mbs{C}$), $\mbb{E}[Q_t]$ is also convex in $\mbs{\mu}$ or $\mbs{\upsilon}$. The constraints on $\mbs{\mu}$ and $\mbs{\upsilon}$ are linear as well. Therefore, the optimization problems~\eqref{opt:DPS},~\eqref{opt:PST} and~\eqref{opt:PPS} are convex programs.
	\end{proof}
	
	\subsection{Proof of Theorem~\ref{theorem:POT}} \label{appdix:POT}
	\begin{proof}
		Following Proposition~\ref{prop:CVX} that~\eqref{opt:DPS} is a convex program with all affine constraints, we first show strong duality based on Slater's condition~\cite[Ch.~5.2.3]{boyd2004convex} that there always exists a feasible point. We then perturb $\rho$ and $(\epsilon_s,\epsilon_t)$ in the original problem and argue about the changes in the optimal value $\mbb{E}^*[Q_t]$ by sensitivity analysis~\cite[Ch.~5.6]{boyd2004convex}. 
		
		For $\epsilon\geq 0$, any rank-one channel matrix $\mbs{C}=\mbf{1}\cdot\mbs{\mu}^\top$ with arbitrary probability vector $\mbs{\mu}$ satisfies the privacy and right stochastic constraints. Then the first 2 constraints in~\eqref{opt:DPS} reduce to,
		\begin{align}
			\mbs{\lambda}^\top\mbs{a} < \mbs{\mu}^\top\mbs{d} \leq \mbs{\lambda}^\top\mbs{a} / \rho.
		\end{align}
		It's easy to see that we can always find a probability vector $\mbs{\mu}$ such that $\mbs{\mu}^\top\mbs{d}=\mbs{\lambda}^\top\mbs{a} / \rho\in(\mbs{\lambda}^\top\mbs{a}, d_n]$ for $\rho\in[\mbs{\lambda}^\top\mbs{a}/d_n, 1)$ with $d_n\geq a_n$ and $d_0=0$. Therefore, feasibility hence strong duality holds for the convex program~\eqref{opt:DPS}.
		
		To simplify the sensitivity analysis, we let $\epsilon_s,\frac{\epsilon_t}{2}=\epsilon$ w.l.o.g. in the set of constraints~\eqref{eqn:priv_constr_alt} which reduces to 
		\begin{align}
			c_{ij}-e^\epsilon c_{kj}\leq 0,\ \forall (i,k,j)\in N^2\times M.
		\end{align}
		The other constraints in~\ref{opt:DPS} not involving $\epsilon$ and $\rho$ must always hold, but in the sequel we do not make them explicit wherever applicable. 
		
		Let $g(\alpha,\mbs{\beta})$ be the Lagrange dual function and $\mc{F}$ be the feasible set of the original unperturbed problem~\eqref{opt:DPS},
		\begin{align}
			g(\alpha,\mbs{\beta}) 
			= \inf_{\mbs{C}\in\mc{F}} 
			\Bigg( 
			&   \mbb{E}_{\mbs{C}}[Q_t] 
			+ \alpha \left(
			\mbs{\lambda}^\top\mbs{C}\mbs{d}\cdot\rho - \mbs{\lambda}^\top\mbs{a}
			\right) \nonumber \\
			&   + \sum_{r=1}^{R_\epsilon}\beta_r(c_{ij}-e^\epsilon c_{kj})
			\Bigg), \label{eqn:dual_func}
		\end{align}
		where $R_\epsilon=|N^2\times M|$ and w.l.o.g. we omit the terms corresponding to the other constraint functions. Denote $(\alpha^*,\mbs{\beta}^*)\succeq0$ as the optimal dual variables. Let $p^*(\Delta\rho, \Delta\epsilon)$ be the optimal value of the perturbed problem,
		\begin{equation}
			\begin{aligned} 
				\min_{\mbs{C}}\quad & \mbb{E}_{\mbs{C}}[Q_t] \\
				\subjto
				\quad & \mbs{\lambda}^\top\mbs{C}\mbs{d}\cdot(\rho+\Delta\rho) \leq \mbs{\lambda}^\top\mbs{a}, \\
				& c_{ij}-e^{\epsilon+\Delta\epsilon}c_{kj}\leq 0,\ \forall (i,k,j)\in N^2\times M.
			\end{aligned}
			\tag{$\tilde{\mc{P}}^\textrm{DPS}$}\label{opt:DPS_perturbed}
		\end{equation}
		Then we have, by strong duality, for any feasible point $\tilde{\mbs{C}}$ in the perturbed problem~\eqref{opt:DPS_perturbed},
		\begin{align}
			p^*(0,0) 
			&   = g(\alpha^*,\mbs{\beta}^*)
			\leq \mbb{E}_{\tilde{\mbs{C}}}[Q_t] - \alpha^* \mbs{\lambda}^\top\tilde{\mbs{C}}\mbs{d} \cdot \Delta\rho \label{eqn:delta_rho}\\
			&   + \sum_{r=1}^{R_\epsilon}\beta^*_r
			\left(
			\left(
			\tilde{c}_{ij}-e^{\epsilon+\Delta\epsilon}\tilde{c}_{kj}
			\right) 
			+ e^\epsilon \tilde{c}_{kj} \left(e^{\Delta \epsilon}-1\right)
			\right), \nonumber
		\end{align}
		rearrange the terms and further simplify,
		\begin{align}
			\mbb{E}_{\tilde{\mbs{C}}}[Q_t] \geq p^*(0,0) 
			- \sum_{r=1}^{R_\epsilon}\beta^*_r e^\epsilon \left(e^{\Delta \epsilon}-1\right) 
			+ \alpha^*\mbs{\lambda}^\top\mbs{a} \cdot \Delta\rho. \label{eqn:sensitivity_analysis}
		\end{align}
		The inequality in~\eqref{eqn:delta_rho} follows from~\eqref{eqn:dual_func} and the constraint $\mbs{\lambda}^\top\tilde{\mbs{C}}\mbs{d}\cdot(\rho+\Delta\rho) \leq \mbs{\lambda}^\top\mbs{a}$ for feasible $\tilde{\mbs{C}}$. Likewise, the simplification to the inequality~\eqref{eqn:sensitivity_analysis} results from the constraint $\mbs{\lambda}^\top\tilde{\mbs{C}}\mbs{d}\geq\mbs{\lambda}^\top\mbs{a}$ and $\tilde{c}_{kj}\leq 1$.
		
		In the perturbed problem,~\eqref{eqn:sensitivity_analysis} holds for any feasible $\tilde{\mbs{C}}$, so does the optimal $\tilde{\mbs{C}}^*$. Meanwhile, $p^*(0,0)$ is the optimal value of the original unperturbed problem. Thereby,
		\begin{itemize}
			\item If we increase the privacy guarantee by decreasing $\epsilon$ ($\Delta\epsilon<0$) with $\rho$ fixed, then the minimum expected queue size across time $\mbb{E}^*[Q_t]$ increases.
			\item If we increase the transmission efficiency level $\rho$ ($\Delta\rho>0$) with $\epsilon$ fixed, then $\mbb{E}^*[Q_t]$ increases as well.
		\end{itemize}
	\end{proof}
	One may want to show the convexity of $p^*(\Delta\rho,\Delta\epsilon)$ in $\Delta\rho$ or $\Delta\epsilon$ following the standard procedure in sensitivity analysis. This is false, however, as the support of $p^*(\Delta\rho,\Delta\epsilon)$ is not convex in either $\Delta\rho$ or $\Delta\epsilon$. We omit the proof due to space limit. The intuition is that when we convert~\eqref{opt:DPS_perturbed} to its standard form, $\Delta\rho$ and $\Delta\epsilon$ do not appear on the right hand side of the inequality constraints.
	
	\section*{Acknowledgment}
	Special thanks to Noah Apthorpe, Dillon Reisman and Nick Feamster for sharing the network traffic data of smart home IoT devices.
	
	\ifCLASSOPTIONcaptionsoff
	\newpage
	\fi

	
	
	%
		
		
	
	\bibliographystyle{IEEEtran}
	\bibliography{ToN}

\begin{thebibliography}{10}
\providecommand{\url}[1]{#1}
\csname url@samestyle\endcsname
\providecommand{\newblock}{\relax}
\providecommand{\bibinfo}[2]{#2}
\providecommand{\BIBentrySTDinterwordspacing}{\spaceskip=0pt\relax}
\providecommand{\BIBentryALTinterwordstretchfactor}{4}
\providecommand{\BIBentryALTinterwordspacing}{\spaceskip=\fontdimen2\font plus
\BIBentryALTinterwordstretchfactor\fontdimen3\font minus
  \fontdimen4\font\relax}
\providecommand{\BIBforeignlanguage}[2]{{%
\expandafter\ifx\csname l@#1\endcsname\relax
\typeout{** WARNING: IEEEtran.bst: No hyphenation pattern has been}%
\typeout{** loaded for the language `#1'. Using the pattern for}%
\typeout{** the default language instead.}%
\else
\language=\csname l@#1\endcsname
\fi
#2}}
\providecommand{\BIBdecl}{\relax}
\BIBdecl

\bibitem{porambage2016quest}
P.~Porambage, M.~Ylianttila, C.~Schmitt, P.~Kumar, A.~Gurtov, and A.~V.
  Vasilakos, ``{The Quest for Privacy in the Internet of Things},'' \emph{IEEE
  Cloud Computing}, vol.~3, no.~2, pp. 36--45, 2016.

\bibitem{lu2018internet}
Y.~Lu and L.~Da~Xu, ``{Internet of Things (IoT) Cybersecurity Research: A
  Review of Current Research Topics},'' \emph{IEEE Internet of Things Journal},
  vol.~6, no.~2, pp. 2103--2115, 2018.

\bibitem{raval2019olympus}
N.~Raval, A.~Machanavajjhala, and J.~Pan, ``{Olympus: Sensor Privacy through
  Utility Aware Obfuscation},'' \emph{Proceedings on Privacy Enhancing
  Technologies}, vol. 2019, no.~1, pp. 5--25, 2019.

\bibitem{malekzadeh2020privacy}
M.~Malekzadeh, R.~G. Clegg, A.~Cavallaro, and H.~Haddadi, ``{Privacy and
  Utility Preserving Sensor-Data Transformations},'' \emph{Pervasive and Mobile
  Computing}, p. 101132, 2020.

\bibitem{xiong2016randomized}
S.~Xiong, A.~D. Sarwate, and N.~B. Mandayam, ``{Randomized Requantization with
  Local Differential Privacy},'' in \emph{Acoustics, Speech and Signal
  Processing (ICASSP), 2016 IEEE International Conference on}.\hskip 1em plus
  0.5em minus 0.4em\relax IEEE, 2016, pp. 2189--2193.

\bibitem{tahaei2020rise}
H.~Tahaei, F.~Afifi, A.~Asemi, F.~Zaki, and N.~B. Anuar, ``{The Rise of Traffic
  Classification in IoT Networks: A Survey},'' \emph{Journal of Network and
  Computer Applications}, vol. 154, p. 102538, 2020.

\bibitem{apthorpe2017closing}
N.~Apthorpe, D.~Reisman, and N.~Feamster, ``{Closing the Blinds: Four
  Strategies for Protecting Smart Home Privacy from Network Observers},''
  \emph{arXiv preprint arXiv:1705.06809}, 2017.

\bibitem{das2016uncovering}
A.~K. Das, P.~H. Pathak, C.-N. Chuah, and P.~Mohapatra, ``{Uncovering Privacy
  Leakage in BLE Network Traffic of Wearable Fitness Trackers},'' in
  \emph{Proceedings of the 17th International Workshop on Mobile Computing
  Systems and Applications}.\hskip 1em plus 0.5em minus 0.4em\relax ACM, 2016,
  pp. 99--104.

\bibitem{buttyan2012traffic}
L.~Buttyan and T.~Holczer, ``{Traffic Analysis Attacks and Countermeasures in
  Wireless Body Area Sensor Networks},'' in \emph{World of Wireless, Mobile and
  Multimedia Networks (WoWMoM), 2012 IEEE International Symposium on a}.\hskip
  1em plus 0.5em minus 0.4em\relax IEEE, 2012, pp. 1--6.

\bibitem{seliem2018towards}
M.~Seliem, K.~Elgazzar, and K.~Khalil, ``{Towards Privacy Preserving IoT
  Environments: A Survey},'' \emph{Wireless Communications and Mobile
  Computing}, vol. 2018, 2018.

\bibitem{chaum1981untraceable}
D.~L. Chaum, ``{Untraceable Electronic Mail, Return Addresses, and Digital
  Pseudonyms},'' \emph{Communications of the ACM}, vol.~24, no.~2, pp. 84--90,
  1981.

\bibitem{kesdogan1998stop}
D.~Kesdogan, J.~Egner, and R.~B{\"u}schkes, ``{Stop-and-Go-Mixes Providing
  Probabilistic Anonymity in an Open System},'' in \emph{International Workshop
  on Information Hiding}.\hskip 1em plus 0.5em minus 0.4em\relax Springer,
  1998, pp. 83--98.

\bibitem{fu2003analytical}
X.~Fu, B.~Graham, R.~Bettati, W.~Zhao, and D.~Xuan, ``{Analytical and Empirical
  Analysis of Countermeasures to Traffic Analysis Attacks},'' in \emph{2003
  International Conference on Parallel Processing, 2003. Proceedings.}\hskip
  1em plus 0.5em minus 0.4em\relax IEEE, 2003, pp. 483--492.

\bibitem{wang2008dependent}
W.~Wang, M.~Motani, and V.~Srinivasan, ``{Dependent Link Padding Algorithms for
  Low Latency Anonymity Systems},'' in \emph{Proceedings of the 15th ACM
  conference on Computer and communications security}.\hskip 1em plus 0.5em
  minus 0.4em\relax ACM, 2008, pp. 323--332.

\bibitem{wright2009traffic}
C.~V. Wright, S.~E. Coull, and F.~Monrose, ``{Traffic Morphing: An Efficient
  Defense Against Statistical Traffic Analysis},'' in \emph{NDSS}, vol.~9,
  2009.

\bibitem{iacovazzi2014internet}
A.~Iacovazzi and A.~Baiocchi, ``{Internet Traffic Privacy Enhancement with
  Masking: Optimization and Tradeoffs},'' \emph{IEEE Transactions on Parallel
  and Distributed Systems}, vol.~25, no.~2, pp. 353--362, 2014.

\bibitem{apthorpe2019keeping}
N.~Apthorpe, D.~Y. Huang, D.~Reisman, A.~Narayanan, and N.~Feamster, ``{Keeping
  the Smart Home Private with Smart(er) IoT Traffic Shaping},''
  \emph{Proceedings on Privacy Enhancing Technologies}, vol. 2019, no.~3, pp.
  128--148, 2019.

\bibitem{alshehri2020attacking}
A.~Alshehri, J.~Granley, and C.~Yue, ``{Attacking and Protecting Tunneled
  Traffic of Smart Home Devices},'' in \emph{Proceedings of the Tenth ACM
  Conference on Data and Application Security and Privacy}, 2020, pp. 259--270.

\bibitem{xiong2018defending}
S.~Xiong, A.~D. Sarwate, and N.~B. Mandayam, ``{Defending Against Packet-Size
  Side-Channel Attacks in IoT Networks},'' in \emph{2018 IEEE International
  Conference on Acoustics, Speech and Signal Processing (ICASSP)}.\hskip 1em
  plus 0.5em minus 0.4em\relax IEEE, 2018, pp. 2027--2031.

\bibitem{diaz2004taxonomy}
C.~Diaz and B.~Preneel, ``{Taxonomy of Mixes and Dummy Traffic},'' in
  \emph{Information Security Management, Education and Privacy}.\hskip 1em plus
  0.5em minus 0.4em\relax Springer, 2004, pp. 217--232.

\bibitem{danezis2004traffic}
G.~Danezis, ``{The Traffic Analysis of Continuous-Time Mixes},'' in
  \emph{International Workshop on Privacy Enhancing Technologies}.\hskip 1em
  plus 0.5em minus 0.4em\relax Springer, 2004, pp. 35--50.

\bibitem{sun2002statistical}
Q.~Sun, D.~R. Simon, Y.-M. Wang, W.~Russell, V.~N. Padmanabhan, and L.~Qiu,
  ``{Statistical Identification of Encrypted Web Browsing Traffic},'' in
  \emph{Proceedings 2002 IEEE Symposium on Security and Privacy}.\hskip 1em
  plus 0.5em minus 0.4em\relax IEEE, 2002, pp. 19--30.

\bibitem{wright2007language}
C.~V. Wright, L.~Ballard, F.~Monrose, and G.~M. Masson, ``{Language
  Identification of Encrypted VoIP Traffic: Alejandra y Roberto or Alice and
  Bob?}'' in \emph{USENIX Security Symposium}, vol.~3, 2007, pp. 43--54.

\bibitem{white2011phonotactic}
A.~M. White, A.~R. Matthews, K.~Z. Snow, and F.~Monrose, ``{Phonotactic
  Reconstruction of Encrypted VoIP Conversations: Hookt on Fon-iks},'' in
  \emph{2011 IEEE Symposium on Security and Privacy}.\hskip 1em plus 0.5em
  minus 0.4em\relax IEEE, 2011, pp. 3--18.

\bibitem{feghhi2016proportional}
S.~Feghhi, D.~J. Leith, and M.~Karzand, ``{Proportional Fair Rate Allocation
  for Private Shared Networks},'' in \emph{2016 IEEE Symposium on Computers and
  Communication (ISCC)}.\hskip 1em plus 0.5em minus 0.4em\relax IEEE, 2016, pp.
  1006--1011.

\bibitem{mathur2011bit}
S.~Mathur and W.~Trappe, ``{BIT-TRAPS: Building Information-Theoretic Traffic
  Privacy into Packet Streams},'' \emph{IEEE Transactions on Information
  Forensics and Security}, vol.~6, no.~3, pp. 752--762, 2011.

\bibitem{dwork2006our}
C.~Dwork, K.~Kenthapadi, F.~McSherry, I.~Mironov, and M.~Naor, ``{Our Data,
  Ourselves: Privacy via Distributed Noise Generation},'' in \emph{Annual
  International Conference on the Theory and Applications of Cryptographic
  Techniques}.\hskip 1em plus 0.5em minus 0.4em\relax Springer, 2006, pp.
  486--503.

\bibitem{dwork2006calibrating}
C.~Dwork, F.~McSherry, K.~Nissim, and A.~Smith, ``{Calibrating Noise to
  Sensitivity in Private Data Analysis},'' in \emph{Theory of cryptography
  conference}.\hskip 1em plus 0.5em minus 0.4em\relax Springer, 2006, pp.
  265--284.

\bibitem{dwork2014algorithmic}
C.~Dwork and A.~Roth, ``{The Algorithmic Foundations of Differential
  Privacy},'' \emph{Foundations and Trends{\textregistered} in Theoretical
  Computer Science}, vol.~9, no. 3--4, pp. 211--407, 2014.

\bibitem{dwork2010differential}
C.~Dwork, M.~Naor, T.~Pitassi, and G.~N. Rothblum, ``{Differential Privacy
  Under Continual Observation},'' in \emph{Proceedings of the forty-second ACM
  symposium on Theory of computing}.\hskip 1em plus 0.5em minus 0.4em\relax
  ACM, 2010, pp. 715--724.

\bibitem{danezis2013measuring}
G.~Danezis, ``{Measuring Anonymity: A Few Thoughts and a Differentially Private
  Bound},'' in \emph{Proceedings of the DIMACS Workshop on Measuring
  Anonymity}, 2013, p.~26.

\bibitem{diaz2004comparison}
C.~Diaz, L.~Sassaman, and E.~Dewitte, ``{Comparison between Two Practical Mix
  Designs},'' in \emph{European Symposium on Research in Computer
  Security}.\hskip 1em plus 0.5em minus 0.4em\relax Springer, 2004, pp.
  141--159.

\bibitem{kasiviswanathan2008note}
S.~P. Kasiviswanathan and A.~Smith, ``{A Note on Differential Privacy: Defining
  Resistance to Arbitrary Side Information},'' 2008.

\bibitem{atzori2010internet}
L.~Atzori, A.~Iera, and G.~Morabito, ``{The Internet of Things: A Survey},''
  \emph{Computer networks}, vol.~54, no.~15, pp. 2787--2805, 2010.

\bibitem{apthorpe2017smart}
N.~Apthorpe, D.~Reisman, and N.~Feamster, ``{A Smart Home is No Castle: Privacy
  Vulnerabilities of Encrypted IoT Traffic},'' \emph{arXiv preprint
  arXiv:1705.06805}, 2017.

\bibitem{subahi2019detecting}
A.~Subahi and G.~Theodorakopoulos, ``{Detecting IoT User Behavior and Sensitive
  Information in Encrypted IoT-App Traffic},'' \emph{Sensors}, vol.~19, no.~21,
  p. 4777, 2019.

\bibitem{lindley1952theory}
D.~V. Lindley, ``{The Theory of Queues with a Single Server},'' in
  \emph{Mathematical Proceedings of the Cambridge Philosophical Society},
  vol.~48, no.~2.\hskip 1em plus 0.5em minus 0.4em\relax Cambridge University
  Press, 1952, pp. 277--289.

\bibitem{gallager2013stochastic}
R.~G. Gallager, \emph{{Stochastic Processes: Theory for Applications}}.\hskip
  1em plus 0.5em minus 0.4em\relax Cambridge University Press, 2013.

\bibitem{bertsekas1992data}
D.~P. Bertsekas, R.~G. Gallager, and P.~Humblet, \emph{{Data Networks}}.\hskip
  1em plus 0.5em minus 0.4em\relax Prentice-Hall International New Jersey,
  1992, vol.~2.

\bibitem{kairouz2015composition}
P.~Kairouz, S.~Oh, and P.~Viswanath, ``{The Composition Theorem for
  Differential Privacy},'' in \emph{International conference on machine
  learning}.\hskip 1em plus 0.5em minus 0.4em\relax PMLR, 2015, pp. 1376--1385.

\bibitem{warner1965randomized}
S.~L. Warner, ``{Randomized Response: A Survey Technique for Eliminating
  Evasive Answer Bias},'' \emph{Journal of the American Statistical
  Association}, vol.~60, no. 309, pp. 63--69, 1965.

\bibitem{kalantari2016optimal}
K.~Kalantari, L.~Sankar, and A.~D. Sarwate, ``{Optimal Differential Privacy
  Mechanisms Under Hamming Distortion for Structured Source Classes},'' in
  \emph{Information Theory (ISIT), 2016 IEEE International Symposium on}.\hskip
  1em plus 0.5em minus 0.4em\relax IEEE, 2016, pp. 2069--2073.

\bibitem{grassmann1989numerical}
W.~K. Grassmann and J.~L. Jain, ``{Numerical Solutions of the Waiting Time
  Distribution and Idle Time Distribution of the Arithmetic GI/G/1 Queue},''
  \emph{Operations Research}, vol.~37, no.~1, pp. 141--150, 1989.

\bibitem{little2008little}
J.~D. Little and S.~C. Graves, ``{Little's Law},'' in \emph{Building
  intuition}.\hskip 1em plus 0.5em minus 0.4em\relax Springer, 2008, pp.
  81--100.

\bibitem{boyd2004convex}
S.~Boyd and L.~Vandenberghe, \emph{{Convex Optimization}}.\hskip 1em plus 0.5em
  minus 0.4em\relax Cambridge university press, 2004.

\bibitem{humblet1982determinism}
P.~A. Humblet, ``{Determinism Minimizes Waiting Time in Queues},'' 1982.

\bibitem{apthorpe2017spying}
N.~Apthorpe, D.~Reisman, S.~Sundaresan, A.~Narayanan, and N.~Feamster,
  ``{Spying on the Smart Home: Privacy Attacks and Defenses on Encrypted IoT
  Traffic},'' \emph{arXiv preprint arXiv:1708.05044}, 2017.

\bibitem{matloff2008introduction}
N.~Matloff, ``{Introduction to Discrete-Event Simulation and the Simpy
  Language},'' \emph{Davis, CA. Dept of Computer Science. University of
  California at Davis. Retrieved on August}, vol.~2, no. 2009, pp. 1--33, 2008.

\bibitem{kellaris2014differentially}
G.~Kellaris, S.~Papadopoulos, X.~Xiao, and D.~Papadias, ``{Differentially
  Private Event Sequences over Infinite Streams},'' 2014.

\bibitem{mcsherry2009privacy}
F.~D. McSherry, ``{Privacy Integrated Queries: An Extensible Platform for
  Privacy-Preserving Data Analysis},'' in \emph{Proceedings of the 2009 ACM
  SIGMOD International Conference on Management of data}.\hskip 1em plus 0.5em
  minus 0.4em\relax ACM, 2009, pp. 19--30.

\bibitem{he2014blowfish}
X.~He, A.~Machanavajjhala, and B.~Ding, ``{Blowfish Privacy: Tuning
  Privacy-Utility Trade-offs Using Policies},'' in \emph{Proceedings of the
  2014 ACM SIGMOD international conference on Management of data}, 2014, pp.
  1447--1458.

\bibitem{song2017pufferfish}
S.~Song, Y.~Wang, and K.~Chaudhuri, ``{Pufferfish Privacy Mechanisms for
  Correlated Data},'' in \emph{Proceedings of the 2017 ACM International
  Conference on Management of Data}, 2017, pp. 1291--1306.

\bibitem{shanthikumar1991strong}
J.~G. Shanthikumar and D.~D. Yao, ``{Strong Stochastic Convexity: Closure
  Properties and Applications},'' \emph{Journal of Applied Probability},
  vol.~28, no.~1, pp. 131--145, 1991.

\end{thebibliography}

	\vfill
	

\end{document}